\documentclass[journal,12pt,onecolumn,draftclsnofoot]{IEEEtran}

\usepackage{cite}
\usepackage{psfrag}
\usepackage{latexsym, color, amsfonts, amssymb}
\usepackage{amsmath}                
\usepackage{amsthm}                 
\usepackage{algorithm, algorithmic}
\usepackage{epstopdf}
\usepackage{mathrsfs, mathtools}
\usepackage{booktabs, multirow}     
\usepackage{subfigure, graphicx}    
\usepackage{balance}
\allowdisplaybreaks[3]              

\theoremstyle{remark}
\newtheorem{definition}{\textbf{Definition}}
\newtheorem{theorem}{\textbf{Theorem}}
\newtheorem{remark}{\textbf{Remark}}

\begin{document}

\title{\Large \bf Differential Private Discrete Noise Adding Mechanism:\\
    Conditions, Properties and Optimization}

\author{
    Shuying Qin$^{\dag}$, Jianping He$^{\dag}$, Chongrong Fang$^{\dag}$, and James Lam$^{\ddag}$
        \thanks{This work was supported in part by the NSF of China under Grants 61973218 and 62103266, and in part by the General Research Fund 17200918.}
        \thanks{${\dag}$: The Dept. of Automation, Shanghai Jiao Tong University, the Key Laboratory of System Control and Information Processing, Ministry of Education of China, and Shanghai Engineering Research Center of Intelligent Control and Management, Shanghai, China. E-mail address: \{QSY-5208, jphe, crfang\}@sjtu.edu.cn. 
        ${\ddag}$: The Dept. of Mechanical Engineering, The University of Hong Kong, Pokfulam Road, Hong Kong E-mail address: james.lam@hku.hk.
        Preliminary results have been accepted by 2022 American Control Conference \cite{my}.}
}

\maketitle

\begin{abstract}
Differential privacy is a standard framework to quantify the privacy loss in the data anonymization process.
To preserve differential privacy, a random noise adding mechanism is widely adopted, where the trade-off between data privacy level and data utility is of great concern.
The privacy and utility properties for the continuous noise adding mechanism have been well studied.
However, the related works are insufficient for the discrete random mechanism on discretely distributed data, e.g., traffic data, health records.
This paper focuses on the discrete random noise adding mechanisms.
We study the basic differential privacy conditions and properties for the general discrete random mechanisms, as well as the trade-off between data privacy and data utility.
Specifically, we derive a sufficient and necessary condition for discrete $\epsilon$-differential privacy and a sufficient condition for discrete $(\epsilon,\delta)$-differential privacy, with the numerical estimation of differential privacy parameters.
These conditions can be applied to analyze the differential privacy properties for the discrete noise adding mechanisms with various kinds of noises.
Then, with the differential privacy guarantees, we propose an optimal discrete $\epsilon$-differential private noise adding mechanism under the utility-maximization framework, where the utility is characterized by the similarity of the statistical properties between the mechanism's input and output.
For this setup, we find that the class of the discrete noise probability distributions in the optimal mechanism is Staircase-shaped.

\end{abstract}


\section{Introduction}

\subsection{Background}
Data anonymization, namely, preventing information from being re-identified \cite{Sam01}, is an important approach to protect data privacy in data publishing.
Random noise adding mechanism is a classic method to achieve data anonymization.
Usually, few random mechanisms can fully protect privacy, i.e., the privacy loss is non-negligible.
To quantify the privacy loss, many privacy frameworks are proposed, include information-theoretic privacy\cite{moulin2003information}, differential privacy\cite{dwork2006calibrating}, and privacy based on secure
multiparty computation\cite{lindell2005secure}, etc.
The frameworks differ mostly in the privacy guarantee strength.
Due to the strong privacy guarantee brought by differential privacy, this  brand-new privacy framework has received wide attention.
It is introduced by Dwork \textit{et al.}\cite{dwork01}, where the idea is inspired by the probabilistic encryption.
The innovation lies in that it is a property towards the anonymization process (i.e., a random noise adding mechanism) rather than the datasets.
Thanks to such a useful property, the differential privacy framework is widely employed in many areas, such as distributed optimization\cite{optimization1, optimization2}, control and network systems\cite{cortes2016differential, katewa2015protecting, han2018privacy, WANG202187}, filtering\cite{field-03, le2020differentially} and others\cite{guo2021topology, field-01, wei2020federated, wang2022consensus}, etc.
Meanwhile, it is able to obtain privacy guarantees and analyze how much information is leaked, e.g., when processing the telemetry data\cite{DKY17}, census data\cite{GAP18}, and medical data\cite{lv2021security}, etc.
Note that the differential privacy preservation achieved by the random noise adding mechanisms is at the cost of the data utility.
Many scholars are dedicated to studying the trade-off between the privacy level and utility for the random noise adding mechanisms.

\subsection{Motivations}
There is a large amount of discrete data in practice, e.g., census records, traffic data, etc.
The data privacy is to be preserved by the random noise adding mechanisms.
To ensure the interpretability of the protected numerical discrete data, the random added noises in the mechanisms should be discretely distributed, i.e., discrete noise adding mechanisms.
When the random noises satisfy the continuous Lipschitz and continuous differentiability, researchers have carried out a series of differential privacy condition studies on the continuous data.
However, the conditions of the Lipschitz continuity and the differentiability are not guaranteed under the discrete scenarios.
It is unclear what would be the problem if the continuous differential privacy conditions were directly applied to the discrete noise adding mechanisms.
Besides, it remains unknown whether the properties for the well-known continuous noise adding mechanisms (e.g., the \textit{Laplacian} and the \textit{Gaussian} mechanisms) can be maintained for the discrete ones, and whether more differential privacy properties are available.
These issues call for discrete privacy-critical studies to ensure the deployment of discrete differential private mechanisms.

Moreover, to improve the utility for the published data, the trade-off between the differential privacy level and utility for the discrete random noise adding mechanisms should be considered.
Most of the existing studies \cite{ghosh2012universally, gupte2010universally, staircase} model the utility function as a general function depending on the noise added to the query output.
This utility measure is reasonable but indirect, since the utility is maximized in terms of minimizing the noises (e.g., the magnitude, the variance).
To make the utility metric more intuitive, one challenge is whether there exists a utility function on the level of distortion after noise addition to the data.
A new insight is given by the similarity degree of the statistical properties between the original data and the noise added data (i.e., the mechanism's input and output).
Given that the inputs and outputs are random variables, the similarity degree can be captured by the probability distributions, a more comprehensive characterisation than the statistical information such as variance and expectation, etc.
Specifically, the degree of similarity is usually quantified by the distance function between the probability distributions \cite{chung1989measures}.
The commonly used distance functions of interest include Kullback-Leibler (KL) divergence\cite{van2014renyi}, Jensen–Shannon (JS) divergence\cite{lin1991divergence}, and Wasserstein distance\cite{panaretos2019statistical}.
In this paper, we adopt the Wasserstein distance as the utility metric.
In contrast to the KL divergence, it satisfies the basic properties of distance (non-negativity, identity of indiscernible, symmetry and triangle inequality).
Compared with the JS divergence, it takes into account the geometric properties between two probability distributions.

With the innovative utility model, it would be desirable to provide an implementable discrete random noise adding mechanism for the utility optimization problem.
However, the explicit expressions for the general form of the Wasserstein distance are rare, except the one-dimensional Gaussian cases.
Besides, even adopting the one-dimensional case directly as an optimization objective, it is still a non-convex optimization problem.
Therefore, it is necessary to find an equivalent form of the primal problem to transform the unconventional optimization problem into a solvable one.

\subsection{Contributions}
Motivated by the above observations, in this paper, we study the differential privacy conditions, properties and utility optimization for the discrete random noise adding mechanisms.
Beginning with the definition of the discrete data, we first clarify the discrete random noise adding mechanism.
As for the differential privacy analysis, we find that the conditions for the discrete differential private mechanism are further simplified compared with the conditions for the continuous mechanism\cite{he2020differential}.
The differential privacy parameters estimation results remain a certain similarity.
Also, compared with the continuous random noise adding mechanism, the differential privacy properties hold well in most discrete scenarios, e.g., the discrete \textit{Gaussian}, \textit{Laplacian} and \textit{Exponential} mechanisms.
Especially, our results for the discrete Gaussian noise adding mechanism are consistent with the literature\cite{canonne2020discrete} to some extent.
More concretely, we obtain the same differential privacy properties and similar differential privacy parameters estimation.

Moreover, as for the trade-off between the privacy level and utility, we select the Wasserstein distance as a new utility measure.
The innovation lies in that the utility model measures the distance between the input and output of the proposed mechanism, by taking the geometric properties of these two discrete distributions into account.
Then, we propose an equivalent form of the primal problem that transforms the non-convex optimization problem into a linear programming problem.
Finally, we obtain the optimal discrete $\epsilon$-differential private mechanism with the Simplex Method \cite{nelder1965simplex}.
In Table \ref{table:work comparison}, we compare various works on the differential privacy and the utility properties for the random noise adding mechanisms.
\begin{table}[t]
\centering
    \caption{Comparison of Works on Mechanism Privacy and Utility}
    \label{table:work comparison}
    \begin{tabular}{c|cccc}
    \toprule
    \multirow{3}{*}{\textbf{\begin{tabular}[c]{@{}c@{}}Differential\\ Privacy\\ Properties\end{tabular}}}
    & \textbf{Works}
        & \textbf{\cite{he2020differential}} & \textbf{\cite{canonne2020discrete}} & \textbf{This work} \\ \cmidrule(l){2-5} 
        & \textbf{Scenario}
            & Continuous & Discrete & Discrete \\ \cmidrule(l){2-5}
        & \textbf{Scope}
            & General & Specific & General \\ \midrule
    \multirow{3}{*}{\textbf{\begin{tabular}[c]{@{}c@{}}Utility\\ Properties\end{tabular}}}
    & \textbf{Works}
        & \textbf{\cite{gupte2010universally}}
            & \textbf{\cite{staircase}}
            & \textbf{This work} \\ \cmidrule(l){2-5} 
        & \textbf{\begin{tabular}[c]{@{}c@{}}Utility\\ Metric\end{tabular}}
            & \multicolumn{2}{c}{\begin{tabular}[c]{@{}c@{}} The minimum \\ noise magnitude/variance \end{tabular}}
            & \begin{tabular}[c]{@{}c@{}}Wasserstein\\ distance\end{tabular}  \\ \cmidrule(l){2-5} 
        & \textbf{\begin{tabular}[c]{@{}c@{}}Optimal\\ mechanism\end{tabular}}
            & \begin{tabular}[c]{@{}c@{}}\textit{Geometric}\\ (Noises related) \end{tabular}
            & \begin{tabular}[c]{@{}c@{}}\textit{Staircase}\\ (Noises.)\end{tabular}
            & \begin{tabular}[c]{@{}c@{}}\textit{Staircase}\\ (Inputs.)\end{tabular} \\ \bottomrule
    \end{tabular}
\end{table}

The differences between this paper and its conference version \cite{my} include
i) the analysis of the differential privacy properties for the discrete Exponential noise adding mechanism,
ii) the optimization of the discrete noise adding mechanisms, i.e., maximizing the data utility under the differential privacy guarantees, 
iii) the sufficient simulations on the optimal discrete noise adding mechanism.

The main contributions are summarized as follows.
\begin{itemize}
    \item \textit{(Conditions.)}
    We investigate general differential privacy conditions for the discrete noise adding mechanisms, 
    i.e., a sufficient and necessary condition for $\epsilon$-differential privacy, and a sufficient one for $(\epsilon,\delta)$-differential privacy.
    Moreover, we obtain a numerical method to estimate the two privacy parameters $\epsilon$ and $\delta$.
    
    \item \textit{(Properties.)}
    We analyze the differential privacy properties and provide the privacy guarantees for the representative discrete noise adding mechanisms with the obtained theories.
    In detail, we investigate the mechanisms under the discrete Gaussian, Laplacian, Staircase-shaped, Uniform, Exponential distributed noises, respectively.
    
    \item \textit{(Optimization.)}
    We study the utility-maximization optimization for the $\epsilon$-differential private mechanisms.
    Defining the utility as the Wasserstein distance between the mechanism input and output probability distributions, we derive an optimal discrete Staircase-shaped noise adding mechanism.
    Further, we conduct extensive simulations to verify its optimality.

\end{itemize}

\subsection{Organization}
The remainder of this paper is organized as follows.
The related works are investigated in Section \ref{sec:Related Work}.
Section \ref{sec:Preliminaries} states necessary preliminaries.
In Section \ref{sec:Main Results}, we give theoretical differential privacy conditions and parameters estimation, perform further analysis on the differential privacy properties, and propose a discrete differential private mechanism with the maximum utility.
Section \ref{sec:Simulation} provides evaluations for the mechanism optimality.
Finally, conclusions are given in Section \ref{sec:Conclusion}.

\section{Related Work} \label{sec:Related Work}
Since Dwork\cite{dwork01} first introduced the \textbf{differential privacy definition} in 2006, it has become the flagship data privacy definition.
Shortly after it was proposed, numerous attack models and different scenarios are adapted to the variants and extensions of the differential privacy \cite{wang2016relation, cuff2016differential}.
More recently, Desfontaines \textit{et al.} \cite{desfontaines2020sok} gave a systematic taxonomy of the existing differential privacy definitions (approximately 225 kinds).
They compared the definitions from seven dimensions, and showed how the new differential privacy definitions are formed with the combination of different dimensions.
This work allowed new practitioners to have a general idea of the differential privacy research area.

The majority differential privacy researches focus on \textbf{the continuous random noise adding mechanisms}, in a bid to achieve anonymity protection for the continuously distributed data.
Regarding the differential privacy analysis for the general continuous random noise adding mechanisms,
He \textit{et al.}\cite{he2020differential} proposed a sufficient and necessary differential privacy condition, with the privacy parameters estimation.
The basic theories can be applied to analyze various random noises.
Then, they performed in-depth analysis on the differential privacy properties, and applied obtained theories on consensus algorithms.
Apart from the related analysis for the general continuous random mechanisms, differential privacy is widely discussed under a specific continuous random noise adding mechanism \cite{mcsherry2007mechanism}.
For instance, the continuous Gaussian noise adding mechanism preserves $(\epsilon,\delta)$-differential privacy for the query functions with infinite dimensions and real values \cite{liu2018generalized}.
Besides, the random mechanism with the continuous Laplacian distributed\cite{dwork01} noise guarantees $\epsilon$-differential privacy.
So far, the differential privacy properties for the continuous random noise adding mechanism have been widely studied, but it is unknown how the results are suitable for the discrete one (adding discrete random noise on discretely distributed data).

Recently, researchers have paid attention to \textbf{the discrete random differential private mechanisms}.
For instance, the \textit{Exponential} mechanism is a well-known discrete noise adding mechanism that guarantees $\epsilon$-differential privacy \cite{mcsherry2007mechanism, dong2020optimal}.
It aims to protect non-numerical discrete data.
In this mechanism, a scoring function is introduced for the query output, and then the final probability of the output is determined by the score.
The analysis of the mechanism is mature, but it does not necessarily apply to the mechanism that protects numerical discrete data.
Furthermore, Canonne \textit{et al.}\cite{canonne2020discrete} studied the differential privacy properties for the discrete Gaussian noise adding mechanism.
They obtained that the discrete \textit{Gaussian} mechanism guarantees essentially the same level of privacy and accuracy as the continuous one.
Apart from the related properties, Koskela \textit{et al.}\cite{koskela2021tight} proposed a Fourier transform based numerical method to compute the differential privacy parameters for discrete-valued mechanisms.
Specifically, they evaluated the privacy loss for the discrete $(\epsilon,\delta)$-differential private mechanisms, and provided the lower and upper $(\epsilon,\delta)$-differential privacy bounds for the subsampled discrete \textit{Gaussian} mechanism.
Despite the excellent properties, we cannot apply the specific conclusions to general discrete noise adding mechanisms, which makes the analysis of differential privacy in distinct scenarios more difficult.

In addition to the extensive research on the differential privacy properties, some works further consider the fundamental \textbf{trade-off between the privacy level and utility for the random noise adding mechanisms}, which are the two vital properties for the mechanisms.
Gupte \textit{et al.}\cite{gupte2010universally} found that  the optimal differential private mechanism is achieved by adding \textit{Geometric} distributed noise, on the basis of a fixed query sensitivity.
Based on the decision theory, they took the information loss caused by the random noise uncertainty as the utility measure.
More concretely, the objective function was to minimize the worst case of the noise variance or the expected magnitude.
In this line of research, \cite{staircase} generalized the fixed sensitivity to an arbitrary value and derived the optimal noise with the \textit{Staircase}-shaped distribution.
The utility metric they adopted was the same as the one in \cite{gupte2010universally}.
This utility model is rational and risk-averse, but it is hard to determine how the added noise affects the original data with this widely-used model directly.
Based on the related work mentioned above, in this paper, we investigate the discrete differential private noise adding mechanism.

\section{Preliminaries} \label{sec:Preliminaries}
In this section, we mainly introduce the discrete random noise adding mechanisms, the differential privacy definition and our utility metric for the random mechanisms.

\subsection{Preliminaries of Discrete Random Mechanisms} \label{subsec:Preliminaries of Discrete Random Mechanisms}
First, we specify the discrete quantitative data discussed in this paper, by introducing a set of discrete numbers with interval $\Delta$, which is given by 
\begin{equation*}
    \mathcal{Z}_{\Delta} = \left\{ {k\left| {k = k_0 \Delta ,k_0 \in \mathbb{Z}} \right.} \right\}, \Delta \in \mathbb{R}^+.
\end{equation*}
Denote $\mathcal{Z}^+_{\Delta}$ as a set of positive numbers in $\mathcal{Z}_\Delta$,
and $\mathcal{Z}^n_{\Delta}$ as a set of $n$-dimensional column vectors $L = \left[ {{x_1},{x_2}, \ldots ,{x_n}} \right]^T$, where $x_i \in \mathcal{Z}_{\Delta}, i\in V=\{1,2,\ldots,n\}$.

With the basic concept of the discrete data $\mathcal{Z}_{\Delta}$, we then introduce a discrete random noise adding mechanism, which is utilized to achieve privacy protection.
Specifically, this mechanism is a randomized function that takes the original data as input and returns an output after adding random noises.
Let $\Omega, \Theta, \mathcal{S} \subseteq \mathcal{Z}_\Delta^n$ represent the $n$-dimensional input, noise, output space, respectively.
Note that $\mathcal{S} \triangleq \Omega \oplus \Theta$, where $\oplus$ refers to the sum of elements with the same dimension.
Then, the general discrete random noise adding mechanism $\mathcal {A}: \Omega \to \mathcal{S}$ is given by
\begin{equation} \label{eq:mechanism_general}
    \mathcal{A}\left( x \right) = x + h\left( \vartheta  \right),
    ~\forall x \in \Omega, \vartheta \in \mathbb{R}^n, h(\vartheta) \in \Theta,
\end{equation}
where $\vartheta = \left[\vartheta_1, \ldots, \vartheta_n\right]^T$ and the output $\mathcal{A}\left( x \right)$ is a $n$-dimensional random variable.
Note that if the added noises are not discretely distributed at initial, then the interpretability and validity of the original discrete data will be destroyed.
To avoid this case, we define the function $h: \mathbb{R}^n \rightarrow \mathcal{Z}_\Delta^n$ as a discretization function, which maps the continuous added noise $\vartheta \in \mathbb{R}^n $ to the discrete one.
In terms of probability distributions, we propose a general discretization method for the variable $\vartheta_i$ in every dimension, which is shown as
\begin{equation} \label{eq:discretize0}
    {p_{\vartheta_i} }\left( k \right) = \int_k^{k + \Delta} {f\left( \vartheta_i \right)d \vartheta_i},~ i\in V,
\end{equation}
where $k\in \mathcal{Z}_{\Delta}$ and $f\left( \vartheta_i \right)$ refers to the probability density function of the original continuous random noise.
The term $p_{\vartheta_i} \left(k\right)$ is the probability of the discretized random variable $\vartheta_i$ when $\vartheta_i=k$.

In summary, the discrete random noise adding mechanism $\mathcal{A}$ represents the process of adding discrete random noise to the discretely distributed data.
To make the expression more concise, we abbreviate it as the discrete random mechanism $\mathcal{A}$.
Further, we distinguish the mechanisms based on the added discrete noise distributions.
For instance, we call the approach to perturb the data by adding discrete Gaussian distributed random noise as the \textit{Gaussian} mechanism.
Similarly, we define the \textit{Laplacian} mechanism, the \textit{Staircase} mechanism, and the \textit{Exponential} mechanism, etc.

\subsection{Background on Differential Privacy}
In this subsection, we introduce the \textit{differential privacy (DP)} properties for the discrete random mechanism $\mathcal{A}$.
In other words, if the mechanism realizes the privacy protection of the numerical discrete data measured by differential privacy, we call it a differential private mechanism.

First, we adopt the adjacency definition to illustrate the protected data.
Consider two $n$-dimensional data that differ only in one dimension.
Our goal is to preserve the privacy of this single record.
That is, we are concerned with the \textit{value} of the record rather than its \textit{presence} in the data.
Based on this observation, we give the definition of $m$-adjacency for two discrete vectors by referring to the studies in \cite{han2016differentially, huang2015differentially, nozari2017differentially, he2020differential}.
\begin{definition} [$m$-adjacency] \label{def:adjacent}
    Given $m \!\in\! \mathcal{Z}_{\Delta}^+$, 
    the pair of vectors $x,y \in {\mathcal{Z}_{\Delta}^n}$ is $m$-adjacent, if for a given $i_0 \in V$, we have
        \begin{equation} \label{adjacent}
            \forall i\in V,~
            {\left| {{x_i} - {y_i}} \right|} \le \left\{ {\begin{aligned}
            &m, &&{i = {i_0}}; \\
            &0, &&{i \ne {i_0}}.
            \end{aligned}} \right.
        \end{equation}
\end{definition}\vspace{2pt}
From (\ref{adjacent}), we obtain that the pair of $m$-adjacent vectors $x$ and $y$ has the same size and differs only in one record with the same dimension, where the difference is no more than $m$.

Next, we present the definition of $(\epsilon,\delta)$-differential privacy for a discrete random mechanism $\mathcal {A}$.
\begin{definition} [$(\epsilon,\delta)$-differential privacy]\label{def:privacy}
    A discrete random mechanism $\mathcal {A}$ is $(\epsilon,\delta)$-DP if for any pairs of $m$-adjacent vectors $x$ and $y$, and for all $\mathcal {O}\subseteq \mathcal{S}$, we have
    \begin{equation} \label{eq:privacy}
        \Pr \left\{ {{\cal A}\left( x \right) \in {\cal O}} \right\} \le {e^\epsilon }\Pr \left\{ {{\cal A}\left( y \right) \in {\cal O}} \right\} + \delta.
    \end{equation}
\end{definition}\vspace{2pt}
Intuitively speaking, a DP mechanism will not reveal more than a bounded amount of information about the data in the probabilistic perspective.
Note that $\epsilon$ and $\delta$ are two key DP parameters.
The positive number $\epsilon$ measures the privacy maintained by the discrete random mechanism.
More concretely, the term $e^\epsilon$ quantifies the privacy loss across the mechanism outputs \cite{DN03}.
With the parameter $\epsilon \rightarrow 0$, the mechanism causes less privacy loss, i.e., achieves better degree of privacy protection.
Moreover, for cases where the upper bound $\epsilon$ does not hold (privacy loss larger than $e^\epsilon$), the parameter $\delta$ functions to compensate for outputs by allowing a small probability of error.
Specifically, if the strong DP property holds ($\delta=0$), we denote $\epsilon$-DP to replace $(\epsilon,0)$-DP for a simplified expression.
By referring to \cite{desfontaines2020sok}, more detailed DP assumptions and explanations are given in Table \ref{table:Preliminaries of DP}.
\begin{table}[ht]
\centering
    \caption{Preliminaries of the \textit{DP} Properties for Mechanisms}
    \label{table:Preliminaries of DP}
    \begin{tabular}{l|l}
    \toprule
    \textbf{Dimension}
        & \textbf{Explanation} \\ \midrule
    \begin{tabular}[c]{@{}l@{}}
        \textbf{Privacy Cost / Loss}\\
        \quad\quad\quad $\epsilon, \delta$
    \end{tabular}
        & \begin{tabular}[c]{@{}l@{}} Qualified by $\epsilon$ and $e^\epsilon$, respectively,\\ allowing a small probability of error $\delta$.\end{tabular} \\ \midrule
    \begin{tabular}[c]{@{}l@{}}
        \textbf{Adjacency Property}\\
        \quad\quad\quad $m$
    \end{tabular}
        & \begin{tabular}[c]{@{}l@{}}Assume datasets have the same size,\\ and differ only in one record. \\  The difference is no more than $m$. \end{tabular} \\ \midrule
    \begin{tabular}[c]{@{}l@{}}
        \textbf{Privacy Level}\\
        \quad\quad\quad expressed in $x,y$
    \end{tabular}
        & \begin{tabular}[c]{@{}l@{}}Associate the data in each dimension\\ with the same acceptable level of risk.\end{tabular} \\ \midrule
    \begin{tabular}[c]{@{}l@{}}
        \textbf{Randomness:} \\
        \quad\quad\quad $\Pr\! \left\{ {{\cal A}(\cdot) \!\in\! {\cal O}} \right\}$
    \end{tabular}
        & \begin{tabular}[c]{@{}l@{}}
            Only comes from the mechanism itself, (i.e., \\the added random noise $\theta$). The input\\ follows a certain probability distribution.
            \end{tabular} \\ \midrule
    \textbf{Computational Power}
        & Assume infinite for attackers. \\ \bottomrule
    \end{tabular}
\end{table}
In the following sections, we analyze the DP properties for any given discrete random mechanism (i.e., $\epsilon$-DP, or $(\epsilon,\delta)$-DP), and give numerical estimation methods for the two DP parameters.

\subsection{Wasserstein Distance}
In the privacy-preserving process implemented by the discrete random mechanism $\mathcal {A}$, in addition to the degree of privacy protection, we also focus on another crucial mechanism property, \textit{utility}, which is characterized by the similarity of the statistical properties of the mechanism's input and output.
In this subsection, we present a general definition of the utility measure, Wasserstein distance.

\begin{definition} [$p$-Wasserstein distance]\label{def:p-WD}
    The $p$-Wasserstein distance between two probability measures $u$ and $v$ on $\mathbb{R}^d$ is
    \begin{align*}
        W^p(u,v)
        = \mathop {\inf }\limits_{X\sim u,Y\sim v} {\left( {\mathbb{E} {{\left\| {X - Y} \right\|}^p}} \right)^{\frac{1}p}},~ p\ge 1,
    \end{align*}
    where $X$ and $Y$ are two $d$-dimensional random vectors with marginals $u$ and $v$.
    The infimum is taken over all joint distributions of the random variables $X$ and $Y$, provided that the $p$-th moments exists.
\end{definition}

Intuitively, the distance $W_p(u, v)$ is the minimal effort required to reconstruct $u$'s mass distribution into the $v$'s.
The effort is quantified by moving every unit of mass from $x$ to $y$ with the cost ${\Vert x-y \Vert}^p$.
In this paper, we focus on the special case of $1$-Wasserstein distance on $\mathbb{R}^{d=1}$.
By referring to \cite{panaretos2019statistical}, the explicit formulae of the Wasserstein distance with $p=1,d=1$ is shown as:
\begin{align} \label{WD-3}
    W^{p=1}(X, Y)
    = \int_{\mathbb{R}}\left|F_{X}(t)-F_{Y}(t)\right| \mathrm{d} t,
\end{align}
where $F_X(\cdot),F_Y(\cdot)$ are the cumulative distribution functions (CDF) of the continuous random variables $X$ and $Y$, respectively.
Further, we extend (\ref{WD-3}) to the discrete situations as the basis of our utility model, with more detailed information illustrated in Section \ref{subsec:The Optimal DP Mechanism}.

Table \ref{table:Primary Notations} summarizes several notations in this paper.
\begin{table}[ht]
\centering 
    \caption{Primary Notations}
    \label{table:Primary Notations}
    \begin{tabular}{l|l}  
        \toprule \\[-8pt]
        \textbf{Notation} & \textbf{Description}\\
        \midrule
        $\Delta$ & The minimum discretization distance\\
        $x,y \in \mathcal{Z}_\Delta^n$ & A pair of $m$-adjacent vectors\\
        $V$ & A set of dimensions, $=\{1,2,\cdots,n\}$\\
        \midrule
        $\Omega \subseteq \mathcal{Z}_\Delta^n$
            & The set of inputs of random mechanisms\\
        $\mathcal{S} \subseteq \mathcal{Z}_\Delta^n$
            & The set of possible outputs of random mechanisms\\
        $\mathcal{O} \subseteq \mathcal{S}$
            & The subset of possible outputs,\\
            & $\mathcal {O}_i$ is a set of $i$-th column element in $\mathcal{O}$, $i \in V$\\
        $\mathcal {A}: \Omega \to \mathcal{S}$\!\!
            & A discrete random mechanism (probabilistic)\\
        $\mathcal{A}(\cdot)$
            & The output of the mechanism $\mathcal {A}$\\
        \midrule
        $\Theta \subseteq \mathcal{Z}_\Delta^n$
            & The set of noises added to the mechanism inputs\\
        $\theta \in \Theta$
            & The noise added to the mechanism input,\\
            & where $\theta_i$ is the $i$-th element of the noise, $i \in V$\\
        \midrule
        $p_{x / \theta / {x+\theta}}$
            & The input / noise / output probability distribution\\
        $p_{x / \theta / {x+\theta}}(\cdot)$
            & The Probability Mass Function (PMF)\\
        ${p_{\theta_i}}\left(k\right)$
            & The probability value of $\theta_i$ at point $k$,\\
            & which is a simplified expression of $Pr(\theta_i = k)$ \\
        $P_{x / \theta / {x+\theta}}(\cdot)$\!\!
            & The Cumulative Distribution Function (CDF)\\
        \bottomrule
    \end{tabular}
\end{table}

\section{Main Results} \label{sec:Main Results}
In this section, we first propose the DP conditions for the discrete random mechanism $\mathcal{A}$, followed by the estimation methods for the DP parameters $\epsilon$ and $\delta$.
Next, we analyze the DP properties for five representative mechanisms.
Then, we consider the trade-off between the privacy level and utility, deriving a $\epsilon$-DP mechanism with the maximum utility.

In this paper, we consider the added noise is discrete by default, i.e., either it has been discretized by the method shown in (\ref{eq:discretize0}) or it is originally discretely distributed.
The simplified discrete random mechanism $\mathcal{A}$ is rewritten as:
\begin{equation} \label{eq:mechanism}
    {\cal A}\left( x \right) = x + \theta,
\end{equation}
where $x\in \Omega\subseteq \mathcal{Z}^n_{\Delta}$, $\theta \in \Theta\subseteq \mathcal{Z}^n_{\Delta}$.
Here, we use $\theta$ to substitute $h\left( \vartheta  \right)$, a function of continuous random variables
\footnote{
Notice that the random noise $\vartheta$ is continuously distributed and we denote $h\left( \vartheta  \right)$ as a discretization process.
The discretization result is a discretely distributed random variable $\theta \in \Theta$, which is the basis for the implementation of the discrete random mechanism $\mathcal{A}$.
To make the mechanism expression more concise, we replace the term $h\left( \vartheta  \right)$ with $\theta$.
Both of them essentially represent the random variables with discrete distributions.
We can use the simplified mechanism (\ref{eq:mechanism}) to analyze the DP properties.
}
mentioned in the general discrete random mechanism (\ref{eq:mechanism_general}).

\subsection{DP Conditions and Parameters Estimation} \label{subsec:DP Conditions}
In this subsection, a sufficient and necessary condition for the $\epsilon$-DP mechanism and a sufficient condition for the $(\epsilon, \delta)$-DP mechanism are given by Theorem \ref{theorem-epsilon} and Theorem \ref{theorem-epsilon-delta-02}, respectively, with the numerical DP parameters estimation.

First of all, we consider the $\epsilon$-DP conditions for the discrete random mechanism $\mathcal {A}$.
\begin{theorem}\label{theorem-epsilon}
    The discrete random mechanism $\mathcal{A}$ satisfies $\epsilon$-DP if and only if (iff) there exists a positive constant $c_b$ such that
    \begin{equation} \label{eq:c2 discrete}
        \underset{\forall m_0\in \left[ -m,m \right],m_0\in \mathcal{Z}_{\Delta}, \forall i\in V}{\mathop{\sup }}\,\frac{{p_{\theta_i}}\left( k - m_0 \right)}{{p_{{\theta_i}}}\left( k \right)}={c_b},
    \end{equation}
    where $k\in \mathcal{Z}_{\Delta}$.
    Moreover, we have that $c_b$ is an increasing function of $m$.
    The privacy parameter $\epsilon$ is estimated by
    \begin{equation} \label{eq:c2 discrete epsilon}
        \epsilon = \log(c_b).
    \end{equation}
\end{theorem}
\begin{proof}
    Please see the proof in the Appendix \ref{proof 01}.
\end{proof}

We make some explanation about the relationship between the privacy cost $\epsilon$ and the adjacency $m$.
Theorem \ref{theorem-epsilon} shows that the privacy loss $e^\epsilon=c_b$ decreases with smaller adjacency $m$ for the pair of two input vectors.
It is consistent with our intuition that the original data with more similarity (smaller $m$) lead to lower privacy loss ($e^\epsilon$), i.e., guaranteeing better privacy (smaller $\epsilon$).

Furthermore, the existence of the least upper bound $c_b$ in (\ref{eq:c2 discrete}) implies that the denominator ${p_{\theta_i}}\left( {k} \right)$ cannot be zero.
For this setup, we obtain a necessary condition for $\epsilon$-DP, i.e.,
\begin{equation} \label{eq:nece}
    \forall i\in V,~ {p_{\theta_i}}\left( {k} \right) > 0,~
    k\in \mathcal{Z}_{\Delta}.
\end{equation}

\begin{remark}
    We further explore the similarities and differences between the discrete DP conditions in Theorem \ref{theorem-epsilon} and the continuous results in \cite{he2020differential}.
    First, the criteria for the discrete random $\epsilon$-DP mechanisms (\ref{eq:c2 discrete}) has the same essence as the continuous ones.
    It means that any adjacent probability ratio for the noise probability distribution should have an upper bound $c_b$.
    With the DP parameter estimation (\ref{eq:c2 discrete epsilon}), it implies that the privacy loss $e^\epsilon$ will not be infinite in the process of protecting any distinct data.
    Meanwhile, the discrete conditions are the simplification of the continuous ones.
    For the continuous random noise distributions, due to the uncountability of the real number set, the potential infinite local maximum and minimum should be considered in any given interval.
    But for the discrete probability distributions, we only need pay attention to whether the probability value at single point is zero (as described in (\ref{eq:nece})).
    The difference shows that the DP parameter $\epsilon$ is highly related to how we discretize a continuous probability distribution.
    It is further explained in Section \ref{subsec:DP Properties}.
\end{remark}

In summary, Theorem \ref{theorem-epsilon} allows us to verify whether a given discrete random mechanism $\mathcal{A}$ is $\epsilon$-DP or not, only relying on the properties of the added discrete noise probability distributions.
This idea is distinguished from the existing work \cite{dwork2014algorithmic}, which validates the DP properties for the mechanisms by the original DP definition.

Next, we consider a more relax notion, $(\epsilon,\delta)$-DP, for cases where the $\epsilon$-DP conditions cannot be strictly met.
In detail, we propose a sufficient condition to verify the $(\epsilon,\delta)$-DP properties for the discrete random mechanism $\mathcal {A}$, along with the estimation of the value of DP parameters $\epsilon$ and $\delta$.

\begin{theorem}\label{theorem-epsilon-delta-02}
    Let $\Theta \subseteq Z^n_\Delta$ be the set of discrete random variable $\theta$.
    Suppose that $\Theta_0$ and $\Theta_1$ are two subsets of $\Theta$, which satisfies 
    $\Theta ={\Theta_0}\bigcup {\Theta_1}$ and ${\Theta _0} \bigcap {\Theta _1} = \emptyset $.
    Assume 
    \begin{equation} \label{eq:sufficient01}
        \sum\nolimits_{\theta \in {\Theta_0}}{{p_{{\theta_i}}}\left( k \right)}\le \delta,
    \end{equation}
    and the condition (\ref{eq:c2 discrete}) holds when $\theta \in {\Theta_1}$, i.e.,
    \begin{align} \label{eq:sufficient02}
        \underset{\forall m_0\in \left[ -m,m \right],m_0\in \mathcal{Z}_{\Delta},\theta \in {\Theta_1}}{\mathop{\sup }}\,\frac{{p_{\theta_i }}\left( k - m_0 \right)}{{p_{{\theta_i}}}\left( k \right)}={c_b},
    \end{align}
    where $\forall i\in V, k\in \mathcal{Z}_{\Delta}$.
    Then the discrete random mechanism $\mathcal{A}$ is $(\epsilon,\delta)$-DP, and the privacy parameter $\epsilon$ is given by
    \begin{equation} \label{eq:sufficient02 epsilon}
        \epsilon = \log \left( {c_b} \right).
    \end{equation}
\end{theorem}
\begin{proof}
    Please see the proof in the Appendix \ref{proof 02}.
\end{proof}

To further verify the rationality of Theorem \ref{theorem-epsilon-delta-02}, we consider the extreme limitations of the two DP parameters $\epsilon$ and $\delta$, according to (\ref{eq:sufficient01}) and (\ref{eq:sufficient02}), respectively:
\begin{itemize}
    \item $\Theta_1 \rightarrow \Theta$ and $\Theta_0 \rightarrow \emptyset$.
    It evolves into the $\epsilon$-DP since
        \begin{equation*}
            \delta = \underset{\Theta_0 \rightarrow \emptyset}{\mathop{\lim }}\,\sum\nolimits_{{\theta \in\Theta_0}}{{p_{{\theta_i}}}\left( k \right)}=0,
        \end{equation*}
    i.e., Theorem \ref{theorem-epsilon} is satisfied.

    \item $\Theta_0 \rightarrow \Theta$ and $\Theta_1 \rightarrow \emptyset$. Then we have 
        \begin{equation*}
            \delta = \underset{\Theta_0 \rightarrow \Theta}{\mathop{\lim }}\,\sum\nolimits_{{\theta \in\Theta_0}}{{p_{{\theta_i}}}\left( k \right)}=1
        \end{equation*}
    and
        \begin{equation*}
            c_b = \underset{\Theta_1 \rightarrow \emptyset}{\mathop{\lim }}\,\underset{\forall m_0\in \left[ -m,m \right],m_0\in \mathcal{Z}_{\Delta},\theta \in {\Theta_1}}{\mathop{\sup }}\,\frac{{p_{\theta_i }}\left( k - m_0\right)}{{p_{{\theta_i}}}\left( k \right)}=1,
        \end{equation*}
    thus $\epsilon = \rm{log}( c_b)=0$.
    Substituting $\epsilon$ and $\delta$ into (\ref{eq:privacy}), we have $\Pr \left\{ {{\cal A}\left( x \right) \in {\cal O}} \right\} 
    \le \Pr \left\{ {{\cal A}\left( y \right) \in {\cal O}} \right\} + 1$.
    Then, one implies that any mechanism $\mathcal{A}$ satisfies $\left( 0,1 \right)$-DP.
        
\end{itemize}

Note that only discussing the general limitations of the DP parameters in the second case is not sufficient, since it can be applied to arbitrary discrete random mechanisms, making the probability error $\delta$ meaningless here.
Thus, it is worth to estimate the tight bound of $\epsilon$ and $\delta$ for every discrete random mechanism, which will be further discussed in Section \ref{subsec:DP Properties}.

Now, we have obtained the conditions for both discrete $\epsilon$-DP and $(\epsilon,\delta)$-DP mechanisms.
The corresponding DP parameters estimation approaches are summarized in Table \ref{table:Summary of the DP Parameters Estimation Methods}.
\begin{table}[ht]
\centering 
    \caption{Summary of the DP Parameters Estimation Methods}
    \label{table:Summary of the DP Parameters Estimation Methods}
    \setlength{\tabcolsep}{4.2mm}
    \begin{tabular}{l|lll}
    \toprule
    \textbf{DP Property}  & \textbf{$c_b$} & \textbf{$\epsilon$} & \textbf{$\delta$} \\ \midrule
    $\epsilon$-DP & Eq. (\ref{eq:c2 discrete}) & Eq. (\ref{eq:c2 discrete epsilon}) & $\delta=0$ \\
    $(\epsilon,\delta)$-DP & Eq. (\ref{eq:sufficient02}) & Eq. (\ref{eq:sufficient02 epsilon}) & Eq. (\ref{eq:sufficient01}) \\ \bottomrule
    \end{tabular}
\end{table}


\subsection{DP Properties and Privacy Guarantees} \label{subsec:DP Properties}
In this subsection, we apply the obtained conditions to discuss the DP properties for two kinds of discrete random mechanisms.
The first kind of mechanism is achieved by adding discrete noises that are discretized from the continuous ones.
Here, four representative mechanisms are selected, i.e., the \textit{Gaussian}, the \textit{Laplacian}, the \textit{Staircase} and the \textit{Uniform} mechanisms.
The second one is obtained through adding noises that are originally discretely distributed.
The most commonly adopted mechanism is the \textit{Exponential} mechanism.
For each mechanism, we derive the DP properties ($\epsilon$-DP or $(\epsilon,\delta)$-DP), followed by the esimated DP parameters based on Theorem \ref{theorem-epsilon} and Theorem \ref{theorem-epsilon-delta-02}.

Recalling the mechanism definition in Section \ref{subsec:Preliminaries of Discrete Random Mechanisms}, we denote every mechanism as the abbreviation of a random mechanism under the corresponding discrete probability distribution.
In other words, the discrete data are anonymously protected by adding the specific kind of discrete random noise.
Since most noises are given by continuous probability density functions (PDF), we first discretize them to obtain the discrete probability mass functions (PMF) based on the proposed discretization methods shown in (\ref{eq:discretize0}).
Note that the following analysis of DP properties is suitable for any discrete random mechanism regardless of the discretization methods.

\textit{1) The Gaussian mechanism:}
\begin{figure}[ht]
    \begin{center}
    \includegraphics[width=0.5\textwidth]{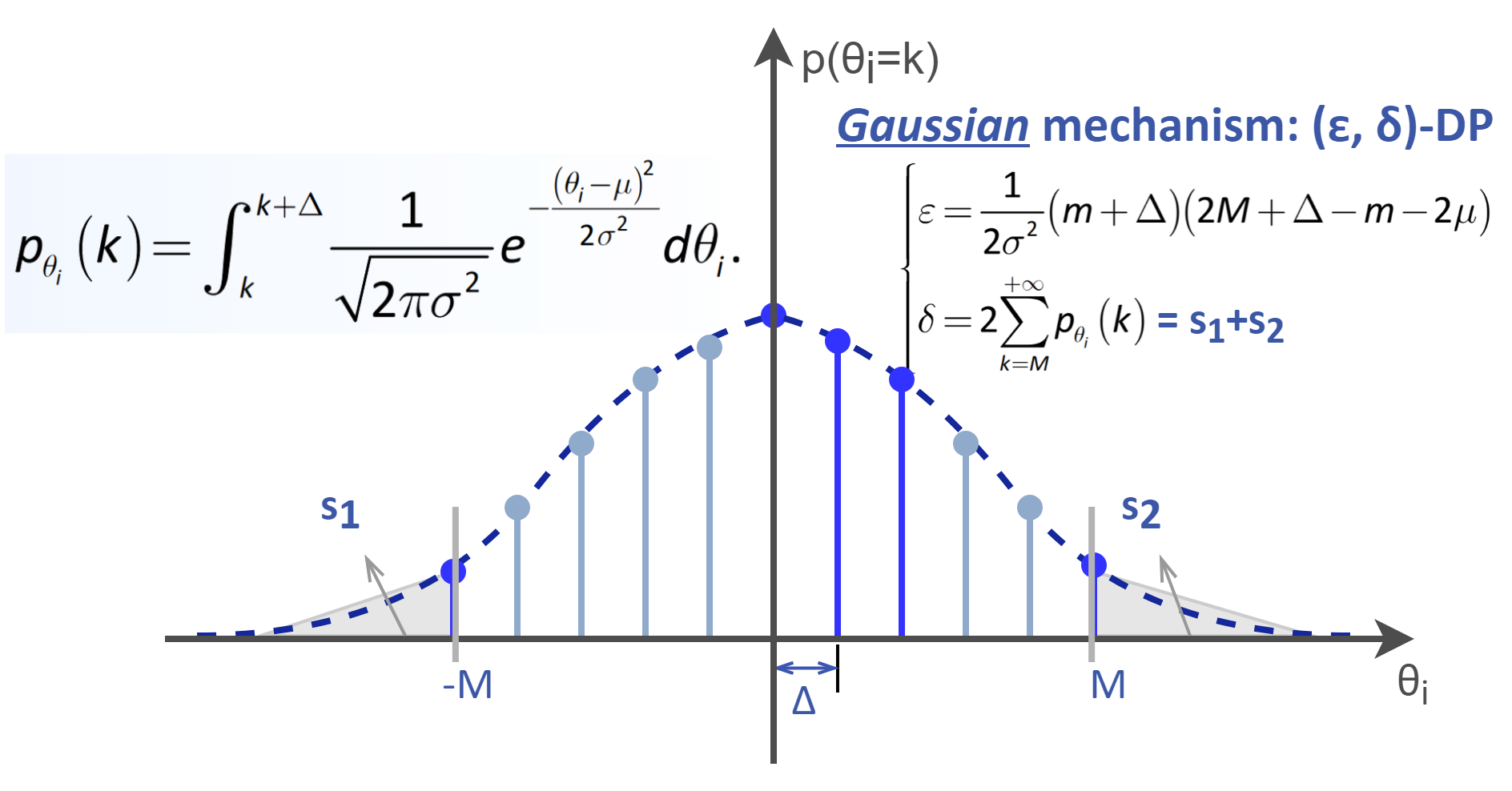}
    \vspace{-10pt}
    \caption{Discrete Gaussian noise distribution.}
    \label{fig:Gaussian noise}
    \end{center}
\end{figure}
This mechanism is realized by adding the discrete Gaussian distributed noise, where the PMF shown in Fig. \ref{fig:Gaussian noise} is given by:
\begin{equation} \label{eq:discrete Gaussian distribution}
    {p_{{\theta _i}}}\left( k \right) = \int_{k}^{k+\Delta} {\frac1{\sqrt {2\pi \sigma^2 }}{e^{ - \frac{{{\left( {\theta_i - \mu} \right)}^2}}{2\sigma^2}}}d\theta_i},~ k\in \mathcal{Z}_{\Delta},
\end{equation}
where the parameters $\mu$ and $\sigma$ are the mean and the standard deviation of the original continuous distribution, respectively.

\begin{theorem}\label{th:Gaussian}
    The \textit{Gaussian} mechanism $\mathcal{A}$ is $(\epsilon,\delta)$-DP.
    Given an arbitrary large constant $M$, the two DP parameters $\epsilon$ and $\delta$ are estimated by
    \begin{equation} \label{eq:result-Gaussian-1}
        \epsilon \!\!=\!\!\!\!\! \underset{m_0\in \left[ -m,m \right],m_0\in \mathcal{Z}_{\Delta}}{\mathop{\sup }}\,\!\! \frac1{2{\sigma^2}}\left( {\left| m_0\right| \!+\! \Delta} \right)\left(2M\!+\!\Delta\!-\!m_0\!-\!2\mu \right)\!\!
    \end{equation}
    and
    \begin{equation} \label{eq:result-Gaussian-2}
        \delta = \mathop {\max }\limits_{i \in V} \left( \frac1{\sqrt {2\pi \sigma^2}}\sum\nolimits_{\Phi_i} {\int_{k}^{k+\Delta} {{e^{ - \frac{{{{\left( {\theta_i - \mu} \right)}^2}}}{2{\sigma^2}}}}d\theta_i} } \right),
    \end{equation}
    where $\Phi _i = \left( { - \infty , - M} \right] \cup \left[ {M,\infty } \right) \subseteq \mathcal{Z}_\Delta$.
\end{theorem}
\begin{proof}
    Please see the proof in the Appendix \ref{proof 03}.
\end{proof}

\begin{remark}
    We find that applying the DP conditions for the general discrete random mechanisms to the analysis of the specific \textit{Gaussian} mechanism yields some similar DP conclusions (Theorem \ref{th:Gaussian}) with those in \cite{canonne2020discrete}.
    First, \cite{canonne2020discrete} also proved that the discrete \textit{Gaussian} mechanism can only provide $\left(\epsilon,\delta \right)$-DP guarantees despite different discretization methods, where the nonzero DP parameter
    \begin{equation*}
        \delta = \mathop {\Pr} \left[ {\theta_i > \frac{{\epsilon {\sigma ^2}}}{\Delta } - \frac{\Delta }2} \right] - {e^\epsilon } \cdot \mathop {\Pr} \left[ {\theta_i > \frac{{\epsilon {\sigma ^2}}}{\Delta } + \frac{\Delta }2} \right]
    \end{equation*}
    determines that it cannot guarantee pure $\epsilon$-DP.
    Then, we consider the similarity of the strict upper bounds on the permissible privacy probability error, $\delta$.
    For instance, when we take the adjacency $m=1$ (the same as the sensitivity in \cite{canonne2020discrete}), we estimate the parameter $\epsilon$ from (\ref{Gaussian epsilon}) by
    \begin{align*}
        \epsilon
        =\!\! \mathop {\max }\limits_{m_0 \in \left\{0,1\right\} } \left[ {\!\frac1{{2{\sigma ^2}}}\left( m_0 + 1 \right)\left( {2M + 1 - {m_0}} \!\right)} \right]
        \!=\! \frac{2M}{\sigma^2}.\!
    \end{align*}
    Substituting it with the estimation approach of $\delta$ proposed in \cite{canonne2020discrete}, i.e., $\delta  \le \frac1{{\sqrt {2\pi {\sigma ^2}} }}{e^{{{ - {{\left\lfloor {\epsilon {\sigma ^2}} \right\rfloor }^2}} \mathord{\left/{\vphantom {{ - {{\left\lfloor {\epsilon {\sigma ^2}} \right\rfloor }^2}} {2{\sigma ^2}}}} \right. \kern-\nulldelimiterspace} {2{\sigma ^2}}}}}$, we have
    \begin{align} \label{his estimation}
        \delta  \le \frac1{{\sqrt {2\pi {\sigma ^2}} }}{e^{ - \frac{2{M^2}}{\sigma ^2}}}.
    \end{align}
    Besides, the estimation value based on our results (\ref{Gaussian delta}) is
    \begin{align} \label{my estimation}
        \delta  \le \frac2{{\sqrt {2\pi {\sigma ^2}} }} \cdot \sum\limits_{k = M}^{ + \infty } \int_{k}^{k+1} {{e^{ - \frac{k^2}{2{\sigma ^2}}}}}.
    \end{align}
    The parameter estimations shown in (\ref{his estimation}) and (\ref{my estimation}) are slightly different, due to the distinct discretization methods.
    Note that the discrete Gaussian distribution in \cite{canonne2020discrete} comes from a natural analogue of the continuous Gaussian, i.e.,
    \begin{equation} \label{eq:discretize_ad}
        p_{\theta_i}(k) =\frac{e^{-(k-\mu)^2 / 2 \sigma^2}}{\sum_{k^{'} \in \mathbb{Z}} e^{-(k^{'}-\mu)^2 / 2 \sigma^2}},
    \end{equation}
    which is a more accurate but complicated discretization method.
    Especially, we find that the DP guarantees for the \textit{Gaussian} mechanisms derived from both discretization methods are almost the same under a large standard deviation $\sigma$.
    Moreover, it is worth mentioning that our DP parameter estimation method (\ref{my estimation}) is more general, as it relies less on the specific noise probability distribution, thanks to the easier discretization approach in (\ref{eq:discretize0}) than (\ref{eq:discretize_ad}).
\end{remark}

\textit{2) The Laplacian mechanism:}
Next, we analyze the DP properties for the \textit{Laplacian} mechanism.
The Laplacian distribution is discretized under (\ref{eq:discretize0}) from the PDF ($f\left( z \right) = \frac1{{2\lambda }}{e^{ - \frac{{\left| {z - \mu } \right|}}{\lambda }}}$).
With simplification, we have
\begin{equation*}
    {p_{{\theta _i}}}\left( {k } \right) = 
    \left\{ {\begin{aligned}
        & {\frac{{1 - {e^{{{-\Delta} \mathord{\left/
        {\vphantom {{ - \Delta } \lambda}} \right.
        \kern-\nulldelimiterspace} \lambda}}}}}2{e^{\frac{{\mu - k }}\lambda}},} &&{k \ge \mu; }\\
        & {\frac{{{e^{{\Delta \mathord{\left/{\vphantom {\Delta \lambda}} \right.
        \kern-\nulldelimiterspace} \lambda}}} - 1}}2{e^{\frac{k - \mu}{\lambda}}},} &&{k \le \mu-\Delta,}
    \end{aligned}} \right.
\end{equation*}
where $k\in \mathcal{Z}_{\Delta}$, $\mu, \lambda$ are the same position and scale parameters as the continuous distribution, respectively.

From Fig. \ref{fig:Laplacian noise}, 
\begin{figure}[ht]
    \begin{center}
    \includegraphics[width=0.5\textwidth]{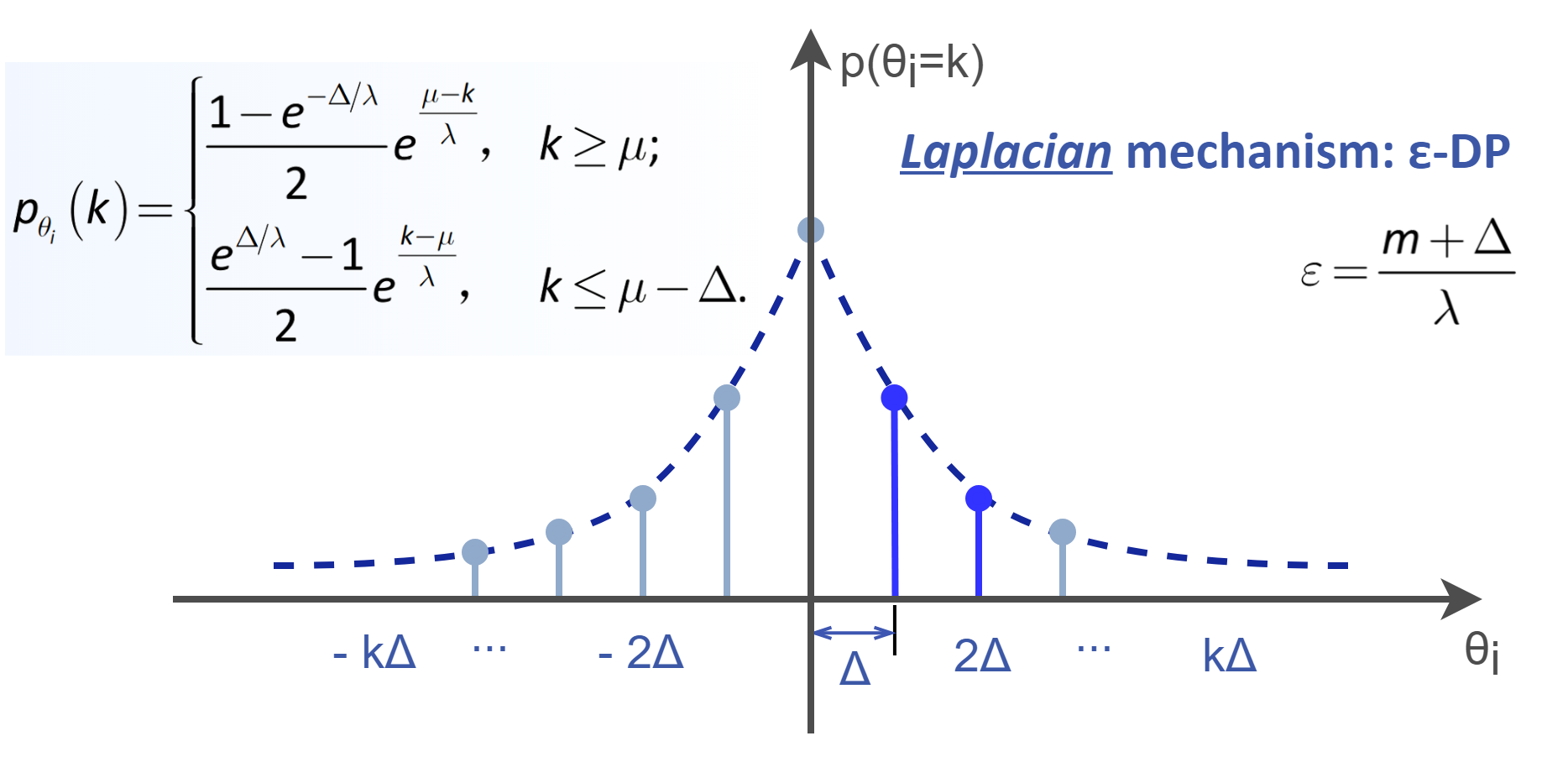}
    \vspace{-10pt}
    \caption{Discrete Laplacian noise distribution.}
    \label{fig:Laplacian noise}
    \end{center}
\end{figure}
it is easy to obtain that the discrete \textit{Laplacian} mechanism guarantees the $\epsilon$-DP precondition in (\ref{eq:nece}).
Then, based on the sufficient and necessary conditions in Theorem \ref{theorem-epsilon}, we find that for any $m_0 \in \left[ { - m,m} \right], m_0\in \mathcal{Z}_{\Delta}$, there exists
\begin{align*}
    \left| {\frac{{{p_{{\theta _i} + m_0 }}\left( {k } \right)}}{{{p_{{\theta _i}}}\left( {k } \right)}}} \right| 
    \le& {e^{{\Delta  \mathord{\left/ {\vphantom {\Delta  \lambda}} \right.
    \kern-\nulldelimiterspace} \lambda}}}\frac{{{e^{ - \frac{{\left| {k  - m_0 - \mu} \right|}}\lambda}}}}{{{e^{ - \frac{{\left| {k  - \mu} \right|}}\lambda}}}}
    = {e^{{\Delta  \mathord{\left/{\vphantom {\Delta  \lambda}} \right.\kern-\nulldelimiterspace} \lambda}}}{e^{\frac{{\left| {k  - \mu} \right|}}\lambda - \frac{{\left| {k  - m_0  - \mu} \right|}}\lambda}} \nonumber\\
    \le& {e^{{\Delta  \mathord{\left/{\vphantom {\Delta  \lambda}} \right.
     \kern-\nulldelimiterspace} \lambda}}}{e^{\frac{m}\lambda}}
    = {e^{\frac{m+\Delta}\lambda}}.
\end{align*}
So we conclude that the \textit{Laplacian} mechanism $\mathcal{A}$ is a discrete $\epsilon$-DP mechanism, where the DP parameter $\epsilon$ is estimated by
\begin{align} \label{eq:result-Laplacian}
    \epsilon  = \log {e^{\frac{ m+\Delta }\lambda}} = \frac{ m+\Delta}\lambda,
\end{align}
which is highly related to the scale parameter $\lambda$.
Then, one implies that the mechanisms' DP properties have strong correlations with the parameters and properties of the specific discrete probability distributions.

\textit{3) The Staircase mechanism:}
Since the study in \cite{staircase} pointed out that the continuous $\epsilon$-DP \textit{Staircase} mechanism performs best in maintaining the data utility, we are interested in this mechanism and hope to verity its DP properties with our general DP conditions.
\begin{figure}[ht]
    \begin{center}
    \includegraphics[width=0.5\textwidth]{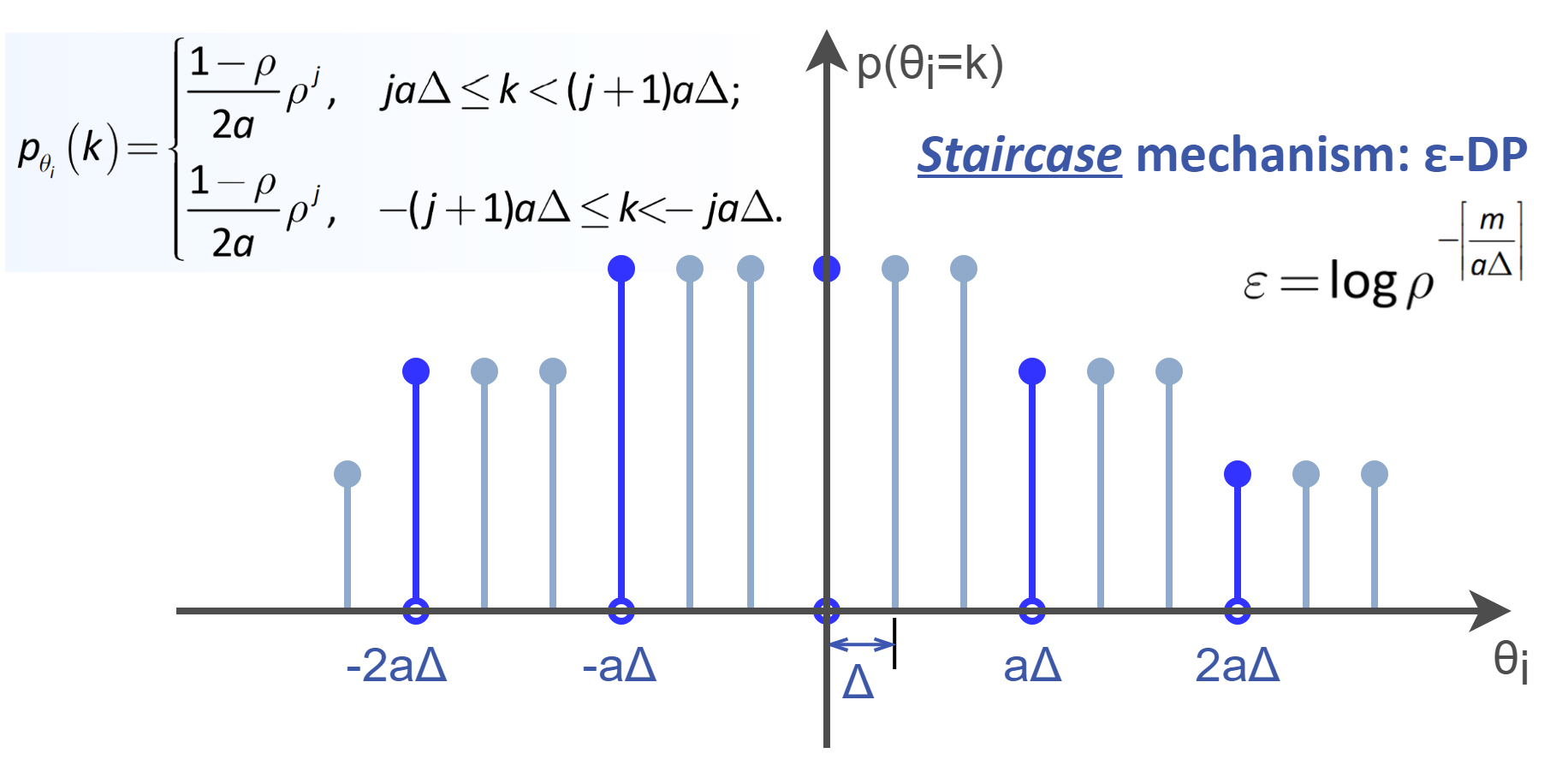}
    \vspace{-10pt}
    \caption{Discrete Staircase-shaped noise distribution.}
    \label{fig:Staircase}
    \end{center}
\end{figure}
First, we derive discrete Staircase-shaped distribution by discretizing the continuous one in \cite{he2020differential}.
The probability distribution in Fig. \ref{fig:Staircase} is obtained by
\begin{equation} \label{stair}
    p_{\theta_i}\left( k \right)
    =\left\{ \begin{aligned}
    &\frac{1-\rho }{2a}{{\rho }^{j}}, && ja \Delta \leq k < (j+1)a \Delta; \\
    &\frac{1-\rho }{2a}{{\rho }^{j}}, && -(j+1)a \Delta \leq k \!<\! -ja  \Delta,\!\!  \\
    \end{aligned} \right.
\end{equation}
where $k\in \mathcal{Z}_{\Delta}, a\in {{\mathbb{Z}}^{+}}, \rho \in \left\{ {x\left| {0 < x < 1,x \in \mathbb{R}} \right.} \right\}, j\in \mathbb{N}$. 
Here $a$ and $\rho$ represent the width and height of the Staircase-shaped distribution, respectively.
It is easy to check (\ref{stair}) as a valid PMF, since
\begin{equation*}
    \sum\limits_{k=-\infty }^{+\infty }{p\left( k \right)}
    =2\sum\limits_{j=0}^{+\infty} a\cdot {\frac{1-\rho }{2a}{{\rho }^j}}=1.
\end{equation*}

Based on Theorem \ref{theorem-epsilon}, we find that the \textit{Staircase} mechanism is $\epsilon$-DP, the same result as \cite{staircase}.
More concretely, for any given adjacency $m$, where $m\in \left\{ {{m}_0}\left| ca\Delta < m_0\le \left( c+1 \right)a\Delta, a,c,\in {{\mathbb{Z}}^{+}} \right. \right\}$, there exists a corresponding upper bound $c_b$ satisfying:
\begin{equation*}
    \underset{\forall m_0\in \left[ -m,m \right],m_0\in \mathcal{Z}_{\Delta}}{\mathop{\sup }}\,\frac{{p_{{\theta_i} }}\left( k-m_0 \right)}{{p_{{\theta_i}}}\left( k \right)}
    =\frac1{{\rho }^{c}}  = c_b.
\end{equation*}
Then, the corresponding DP parameter $\epsilon$ is shown as:
\begin{align} \label{eq:result-Stair}
    \epsilon = \log(c_b)
    = \log {\rho ^{- \left\lceil {\frac{m}{a\Delta}} \right\rceil }},
\end{align}
where the term $\left\lceil {\frac{m}{a\Delta}} \right\rceil$ represents the smallest integer that is not less than $\frac{m}{a\Delta}$.
Furthermore, we have that the \textit{Staircase} mechanism can preserve any given DP levels, with the design of the stair width $a$ and the stair height $\rho$.

\textit{4) The Uniform mechanism:}
The DP properties for the discrete Uniform distributed noise adding mechanism can be easily obtained.
The PMF of the discrete Uniform noise follows:
\begin{align*}
    {p_{{\theta _i}}}\left( {k } \right) = \frac{\Delta}{{b - a + \Delta}},
    ~ a \le {k_0} \le b,
    ~ k_0 \in {\mathcal{Z}_\Delta}.
\end{align*}

Obviously, the \textit{Uniform} mechanism violates the $\epsilon$-DP precondition in (\ref{eq:nece}), since certain probability values are zero (intuitively shown in Fig. \ref{fig:Laplacian noise}), which will lead to infinity privacy loss.
\begin{figure}[ht]
    \begin{center}
    \includegraphics[width=0.5\textwidth]{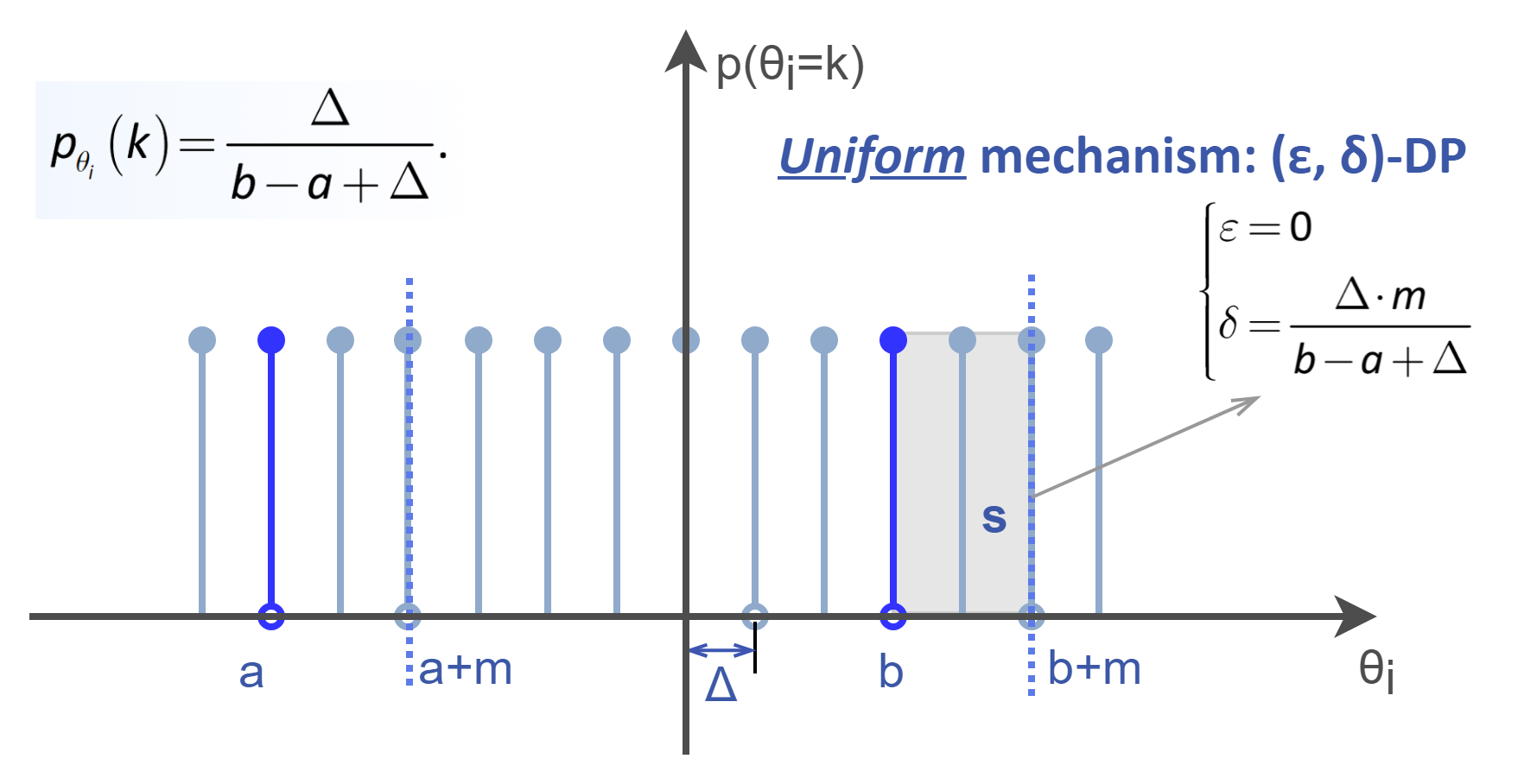}
    \vspace{-10pt}
    \caption{Discrete Unifrom distributed noise distribution.}
    \label{fig:Unifrom}
    \end{center}
\end{figure}
Thus, we have that the \textit{Uniform} mechanism is $(\epsilon,\delta)$-DP.
With the sufficient $(\epsilon,\delta)$-DP condition in Theorem \ref{theorem-epsilon-delta-02}, the DP parameters are given by:
\begin{align} \label{eq:result-Uniform}
    \epsilon=\log(c_b)=\log(1)=0,~ \delta =\frac{m\Delta}{b-a+\Delta}.
\end{align}

\textit{5) The Exponential mechanism:}
Different from the above four discrete random mechanisms, the Exponentially distributed noise added in this mechanism is inherently discrete, i.e., the pre-processing of discretization is not needed.
\begin{figure}[ht]
    \begin{center}
    \includegraphics[width=0.4\textwidth]{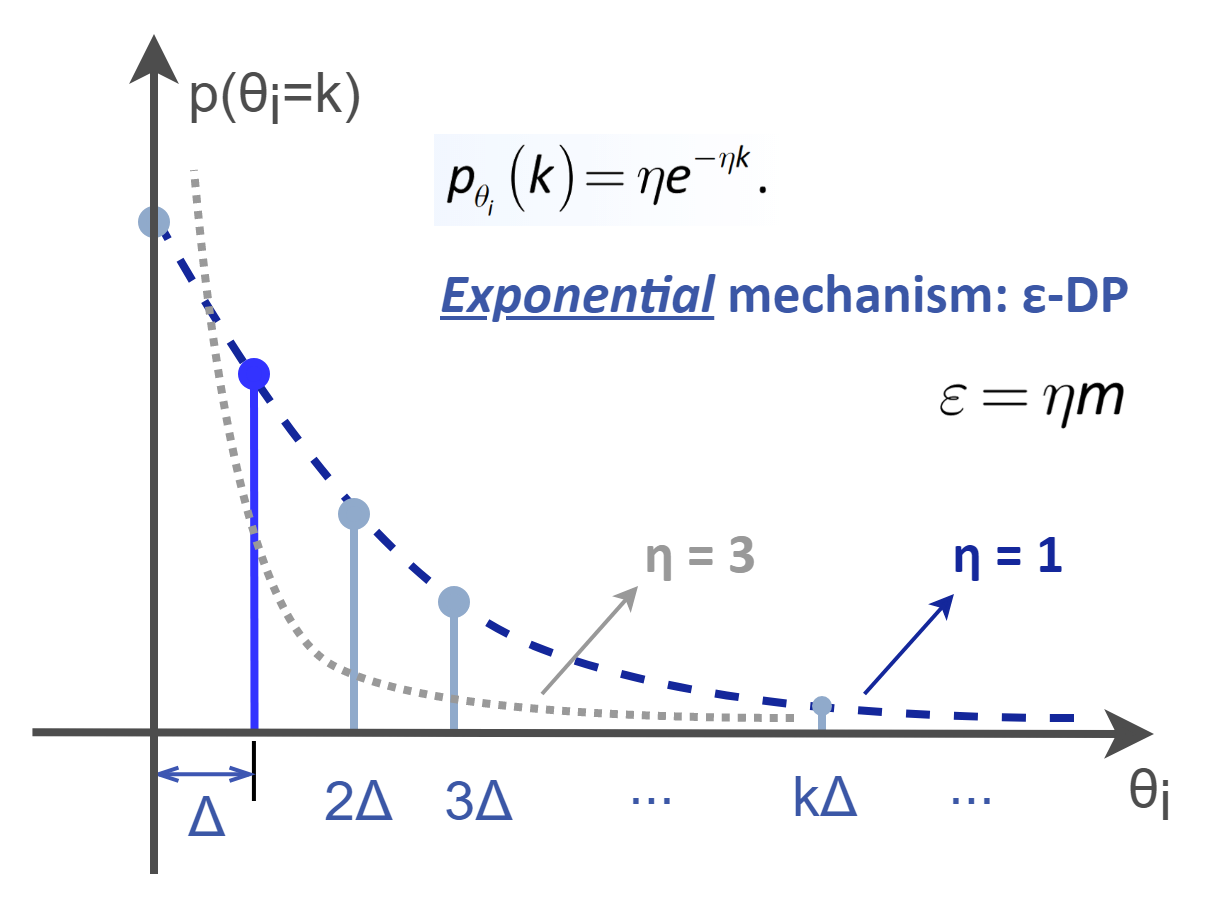}
    \vspace{-10pt}
    \caption{Discrete Exponential noise distribution.}
    \label{fig:Exponential}
    \end{center}
\end{figure}
The PMF of the discrete Exponential noise in Fig. \ref{fig:Exponential} is given by:
\begin{equation} \label{Exponential}
    p_{\theta_i}\left( k \right) = \eta {e^{ - \eta k}},~ k\in \mathcal{Z}_{\Delta}^+,
\end{equation}
where $\eta$ is the rate parameter of the Exponential distribution.

According to Theorem \ref{theorem-epsilon}, we have that the \textit{Exponential} mechanism is $\epsilon$-DP.
Because for any $m_0 \in \left[{- m,m}\right], m_0\in \mathcal{Z}_{\Delta}$ and $k>m_0$, we obtain that
\begin{align*}
    \left| {\frac{{p_{{\theta _i} + m_0 }}\left( {k } \right)}{{p_{\theta _i}}\left( {k } \right)}} \right| 
    = \frac{\eta {e^{ - \eta \left( {k - m_0} \right)}}}{\eta {e^{ - \eta k}}}
    = {e^{\eta m_0}}
    < {e^{\eta m}},
\end{align*}
i.e., the probability ratio is bounded.
Meanwhile, the estimation of the DP parameter $\epsilon$ is 
\begin{equation} \label{eq:result-Exponential}
    \epsilon = \log {e^{\eta m}} = \eta m.
\end{equation}
Thus, we conclude that the privacy cost $\epsilon$ is proportional to the adjacency $m$, which is consistent with the result in \cite{mcsherry2007mechanism}.

\begin{remark}
    Despite the similar DP properties with \cite{mcsherry2007mechanism}, the \textit{Exponential} mechanism we discuss here is slightly different from the existing \textit{Exponential} mechanisms.
    The main reason is that the discrete data we are protecting can be represented numerically.
    The common \textit{Exponential} mechanism protects non-numerical discrete data.
    To achieve differential privacy, the mechanism returns the originally determined result $x$ with a certain probability value, which is highly related to a scoring function $u$.
    The function $u$ gives every result $x$ a score, where a higher score means a higher output probability.
    Formally, the \textit{Exponential} mechanism ${\cal A}_{u,\epsilon}$ with the quality score $u(x)$ and the privacy parameter $\epsilon$ is given by:
    \begin{align*}
       {\cal A}_{u,\epsilon }\left( {x,m} \right) \sim e^{\frac{\epsilon u(x)}{2m}}.
    \end{align*}
    However, the \textit{Exponential} mechanism ${\cal A}$ discussed in this paper aims to protect discretely distributed numerical data.
    Denote the Exponential distribution in (\ref{Exponential}) as $Exp(\eta)$.
    Then, the \textit{Exponential} mechanism $\cal A$ is shown as:
    \begin{equation*}
        {\cal A}\left( x \right) = x + \theta,~ \theta \sim Exp(\eta),
    \end{equation*}
    which is in line with the definition of a general discrete random mechanism $\cal A$ in (\ref{eq:mechanism}). 
    Hence, we can apply the DP conditions in Section \ref{subsec:DP Conditions} to analyze the \textit{Exponential} mechanism.
\end{remark}

In summary, the DP properties for the several typical mechanisms ((\ref{eq:result-Gaussian-1}), (\ref{eq:result-Gaussian-2}), (\ref{eq:result-Laplacian}), (\ref{eq:result-Stair}), (\ref{eq:result-Uniform}), (\ref{eq:result-Exponential}))  are listed in Table \ref{table:discrete differential properties}.
Based on the detailed analysis as well as the comparisons with the existing work, we verify the validity of our conclusions.
\begin{table}[ht]
\centering
    \caption{Discrete DP properties and Privacy Guarantees}
    \label{table:discrete differential properties}
    \begin{tabular}{l|l|c|c}
    \toprule
    \textbf{Mechanism}\!\!
        & \!\!\textbf{Property}
        & \textbf{$\epsilon$}
        & \textbf{$\delta$} \\ \midrule
    \textbf{Gaussian}
        & \!\!$(\epsilon, \delta)$-DP 
        & \!\!\!$\frac1{2{\sigma^2}}\left({m\!+\!\Delta} \right)\left(2M\!+\!\Delta\!-\!m\!-2\mu \right)$\!\!\!
        & \!\!$2\!\sum\limits_{k = M}^{+\infty } {p_{\theta _i}\!\left( k \right)}$\\
    \textbf{Laplacian}
        & \!\!$\epsilon$-DP 
        & $\frac{{m+\Delta}}\lambda$
        & $0$ \\
    \textbf{Staircase}
        & \!\!$\epsilon$-DP
        & $\log {\rho ^{- \left\lceil {\frac{m}{{a\Delta}}} \right\rceil }}$
        & $0$ \\
    \textbf{Uniform}
        & \!\!$(\epsilon, \delta)$-DP
        & $0$
        & $\frac{m\Delta}{b-a+\Delta}$ \\ 
    \textbf{Exponential}
        & \!\!$\epsilon$-DP
        & $\eta m$
        & $0$ \\\bottomrule
    \end{tabular}
    
    \begin{itemize}
        \item Note that the DP properties are related with the discrete noise distributions, the discretization interval $\Delta$, as well as the adjacency $m$.
    \end{itemize}
\end{table}

\subsection{The Optimal DP Mechanism} \label{subsec:The Optimal DP Mechanism}
Based on the two proposed DP conditions for discrete random mechanisms in Section \ref{subsec:DP Conditions}, we further study the trade-off between the privacy level and the utility.
In this subsection, a utility-maximization ( Wasserstein distance-minimization) non-convex optimization problem is formulated, subject to the DP constraints.
Then, we give an equivalent linear programming form of the problem to make it solvable.
At last, we derive an optimal $\epsilon$-DP \textit{Staircase} mechanism.

Since the utility metric is based on the statistical properties of the mechanism's input and output, we begin with the explanation for the inputs and outputs.
Here, we model the input data as a random variable being generated by a specific probability distribution.
Recalling the simplified discrete random mechanism in (\ref{eq:mechanism}), we denote $x, \theta, x+\theta$ as the random variables of the mechanism input, added noise and output, respectively.
From the probabilistic perspective, we believe that the mechanism output is consistent for the same input (shown in (\ref{XY})).
Then, for the original $n$-dimensional data $x$, we only focus on the $i_0$-th dimension data ($x_{i_0}, \theta_{i_0}, x_{i_0}+\theta_{i_0}, {i_0}\in V$) mentioned in (\ref{adjacent}), which need to be protected.
In this problem, for more concise expressions, we omit the subscript $i_0$, i.e., all the three random variables are one-dimensional discrete random variables, with the corresponding PMFs as $p_x(\cdot), p_{\theta}(\cdot), p_{x+\theta}(\cdot): \mathcal{Z}_\Delta \rightarrow \mathbb{R}$, respectively.
Moreover, from the probability theory, the mechanism output is regarded as the summation of two random variables.
The corresponding PMF is computed as:
\begin{align} \label{convolution}
    p_{x + \theta}\left(k \right) = \sum\nolimits_i {{p_x}\left( i \right){p_\theta }\left( {k - i} \right)},~ k,i\in \mathcal{Z}_\Delta.
\end{align}

With the above explanation, we construct the utility optimization problem as following.

\textbf{$\bullet$ Utility model:}
The utility model in this paper is a minimization framework.
We aim to maximize the mechanism utility (or minimize the utility loss) by maintaining the similarity between input and output probability distributions as much as possible, i.e., reducing their Wasserstein distance.
Extending the Wasserstein distance (\ref{WD-3}) based on the continuous distributions, we obtain the discrete Wasserstein distance, a distance function about the discrete random mechanism input and output.
It is formulated as:
\begin{align*}
    W_0(x, x+\theta) = 
    \sum\nolimits_{k \in \mathcal{Z}_\Delta} {\left| {{P_x}\left( k \right) - {P_{x + \theta }}\left( k \right)} \right|},
\end{align*}
where $P(\cdot)$ is the CDF for discrete random variables.
Thus, the objective is to minimize the Wasserstein distance between the input and output, i.e.,
\begin{align} \label{final-objective}
    \mathop {\min}\limits_{{p_\theta }(k)} W_0(x, x+\theta).
\end{align}
The optimization variable is the noise PMF, $p_\theta(\cdot)$.

\textbf{$\bullet$ Constraints:}
The primal constraint comes from the $\epsilon$-DP guarantees of a discrete random mechanism.
Once given the privacy cost $\epsilon$, based on the sufficient and necessary condition in Theorem \ref{theorem-epsilon}, we have
\begin{align} \label{final-constraint-1}
    c_b =\!\! \underset{\forall m_0\in \left[ -m,m \right],m_0\in \mathcal{Z}_{\Delta}}
     {\mathop{\sup }}\!\! \frac{p_{\theta}\left( k-m_0 \right)}{p_{\theta}\left( k \right)} = e^{\epsilon}.
\end{align}
Besides, the existence of the upper bound $c_b$ implies that the probability value at any point should be nonzero.
Combined with the non-negativity of the probability value, we obtain the second constraint:
\begin{align} \label{final-constraint-2}
    {p_\theta }\left( k \right) > 0,\forall k\in \mathcal{Z}_\Delta.
\end{align}
The last constraint is obvious that it should satisfy the basic properties of probability, i.e., the total sum of the probability values should be one:
\begin{align} \label{final-constraint-3}
    \sum\nolimits_{k\in \mathcal{Z}_\Delta} {{p_\theta }(k)} = 1.
\end{align}

\textbf{$\bullet$ Optimization:}
Combining the objective function (\ref{final-objective}) and three constraints (\ref{final-constraint-1})-(\ref{final-constraint-3}), we formulate the following primal optimization problem:
\begin{subequations}\label{P0}
    \begin{align}
        \textbf{P}_\textbf{0}:~~
        \mathop {\min }\limits_{{p_\theta }(k)}~ & W_0(x, x+\theta) \label{P0-a}\\
        \text{s.t.}~~
        & \epsilon \le \log(c_b); \label{P0-b}\\
        & \sum\nolimits_{k\in \mathcal{Z}_\Delta} {{p_\theta }(k)} = 1; \label{P0-c}\\
        & {p_\theta}\left( k \right) \!>\! 0,\forall k\in \mathcal{Z}_\Delta. \label{P0-d}\\\nonumber
    \end{align}
\end{subequations}

Further, to make this non-convex optimization problem $\textbf{P}_\textbf{0}$ solvable, we propose an equivalent problem form.
The main idea is to convert the original problem $\textbf{P}_\textbf{0}$ into a conventional convex optimization problem $\textbf{P}_\textbf{1}$.
First, we introduce two column vectors as the probability distributions, i.e., the input probability distribution
\begin{align} \label{new p_x}
    p_x \!=\!
    \left(\! {\begin{array}{*{20}{c}}
        \dots,\!\!\! &p_{x}(k\!-\!\Delta),\!\!\! &p_{x}(k),\!\! &p_{x}(k\!+\!\Delta),\!\! &\dots
    \end{array}} \!\right)^T
\end{align}
and the noise probability distribution
\begin{align} \label{new p_theta}
    p_{\theta} =
    \left(\! {\begin{array}{*{20}{c}}
        \dots,\!\!\! &p_{\theta}(k\!-\!\Delta),\!\!\! &p_{\theta}(k),\!\! &p_{\theta}(k\!+\!\Delta),\!\! &\dots
    \end{array}} \!\right)^T,
\end{align}
where $k \in \mathcal{Z}_\Delta$ and $p_x(k),p_\theta(k)$ denote the probability value at the point $k$ of the input and the noise, respectively.
Both distributions satisfy the basic probability property, i.e., $\left| p_x \right| \!\!=\!\! 1$ and $\left| p_{\theta} \right| \!\!=\!\! 1$.
With these two notions, we give the equivalent problem $\textbf{P}_\textbf{1}$, where the equivalence is proved in Theorem \ref{theorem-LP}.
\begin{align*}
    \textbf{P}_\textbf{1}:~~
    \mathop {\min }\limits_{p_\theta}~ & W_1(x, x+\theta) \\\nonumber
    \text{s.t.}~~
    & A p_{\theta} \le b,
    ~ \left| p_{\theta} \right| = 1,
    ~ p_{\theta} \succ 0,
\end{align*}
where $W_1(x, x+\theta) = \sum\nolimits_k {\left| {{p_x^T} M_k {p_\theta}} \right|}.$
The term $p_x$ is an arbitrarily given input distribution and $p_\theta$ is the noise distribution that we are interested in.
The matrix $M_k$ and $A,b$ are shown in (\ref{proof:M_k}) and (\ref{proof:AX=b}), respectively.
\vspace{3pt}
\begin{theorem} \label{theorem-LP}
    The problem $\textbf{P}_\textbf{0}$ is equivalent to $\textbf{P}_\textbf{1}$, which means that they have the same optimal solutions.
\end{theorem}
\begin{proof}
    Please see the proof in the Appendix \ref{proof LP}.
\end{proof}

One difficulty of solving the problem $\textbf{P}_\textbf{1}$ is that the optimization variable $p_\theta$ is still involved in the absolute value.
Next, we try to make the optimization variable $p_\theta(k)$ independent of the calculation with the absolute values, which is at the cost of certain results accuracy.
Based on the absolute value inequality, it is easy to have
\begin{align} \label{eq:W1-W2}
    W_1(x, x+\theta) 
    =& \sum\nolimits_k {\left| {{p_x^T} M_k {p_\theta}} \right|} \\\nonumber
    \le& \sum\nolimits_k {\left| {{p_x^T} M_k} \right|} {p_\theta}
    = W_2(x, x+\theta).
\end{align}
Then, we obtain an approximate optimization problem with the standard linear programming form:
\begin{align*}
    \textbf{P}_\textbf{2}:~~
    \mathop {\min }\limits_{p_\theta}~ & W_2(x, x+\theta) \\\nonumber
    \text{s.t.}~~
    & A p_{\theta} \le b,
    ~ \left| p_{\theta} \right| = 1,
    ~ p_{\theta} \succ 0,
\end{align*}
where the objective function $W_2$ is given in (\ref{eq:W1-W2}) and the constraints are the same as the ones in the problem $\textbf{P}_\textbf{1}$.

Finally, we solve the primal optimization problem with $\textbf{P}_\textbf{2}$.
This standard linear programming problem is realized by the \textit{Simplex Method} \cite{nelder1965simplex}.
From $\textbf{P}_\textbf{2}$, one implies that the optimal solution is determined by the input distribution $p_x$ involved in the objective function, as well as the privacy cost $\epsilon$ and the adjacency $m$ in the constraints.
\begin{figure}[ht]
    \begin{center}
    \includegraphics[width=0.5\textwidth]{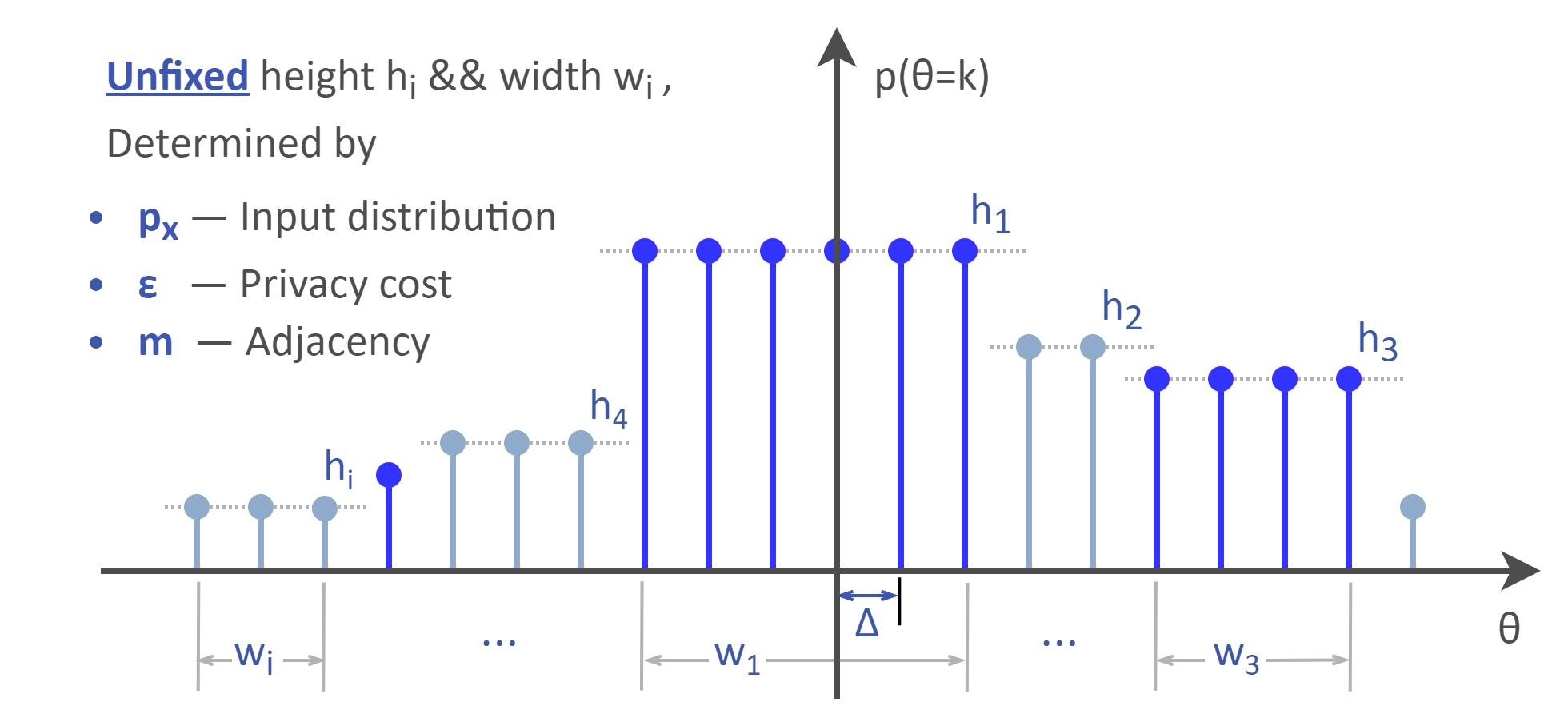}
    \vspace{-5pt}
    \caption{The Staircase-shaped optimal noise distribution.}
    \label{fig:Staircase_opti}
    \end{center}
\end{figure}
Through extensive simulations, we get the optimal discrete random mechanism is realized by the class of \textit{Staircase-shaped} probability distributions, shown in Fig. \ref{fig:Staircase_opti}.
The parameters of the optimal distribution (i.e., the height and the width of the stairs) are unfixed, due to the three factors mentioned above.
The in-depth analysis on how parameters affects the mechanism is provided in Section \ref{sec:Simulation}.

\section{Simulation} \label{sec:Simulation}
In this section, we validate the utility guaranteed by the optimal $\epsilon$-DP \textit{Staircase} mechanism.

\subsection{Simulation Scenario}
First, we give a brief description of the simulation scenario, especially the mechanism inputs and noises.
In this paper, the original discrete data is modeled as a random variable, so we generate the mechanism input with the designated discrete distribution.
Based on the assumptions underlying the mechanism randomness in Table \ref{table:Preliminaries of DP}, we keep the input data constant in each simulation, to guarantee that the mechanism randomness comes only from the added noise.
As for the mechanism noise, the probability distribution is determined by the specific discrete random mechanism.
In this subsection, we consider three representative $\epsilon$-DP mechanisms, i.e., the \textit{Laplacian} mechanism, the \textit{Staircase} mechanism with a lower stair width ($a=5$), and the one with a higher stair width ($a=20$).
Once we set the DP parameter $\epsilon$ and the adjacency parameter $m$, the parameters of each discrete noise distribution can be uniquely determined according to Table \ref{table:Summary of the DP Parameters Estimation Methods}.
Besides, for simplification, we set the minimum discrete interval $\Delta=1$ in the following simulations.

\subsection{The Effect of Three Factors on the Mechanism Utility}
With the above explanation, next, we discuss how the parameters mentioned in Section \ref{subsec:The Optimal DP Mechanism} (the input distribution $p_x$, the privacy cost $\epsilon$, and the adjacency $m$) affect the utility of the optimal $\epsilon$-DP mechanism, respectively.

\begin{itemize}
    \item The effect of the input distribution $p_x$
\end{itemize}
\setcounter{subfigure}{0}
\begin{figure}
    \centering
    \subfigure[Optimal distribution for the Gaussian input.]{\label{fig:opti Gaussian}
        \includegraphics[width=0.45\textwidth]{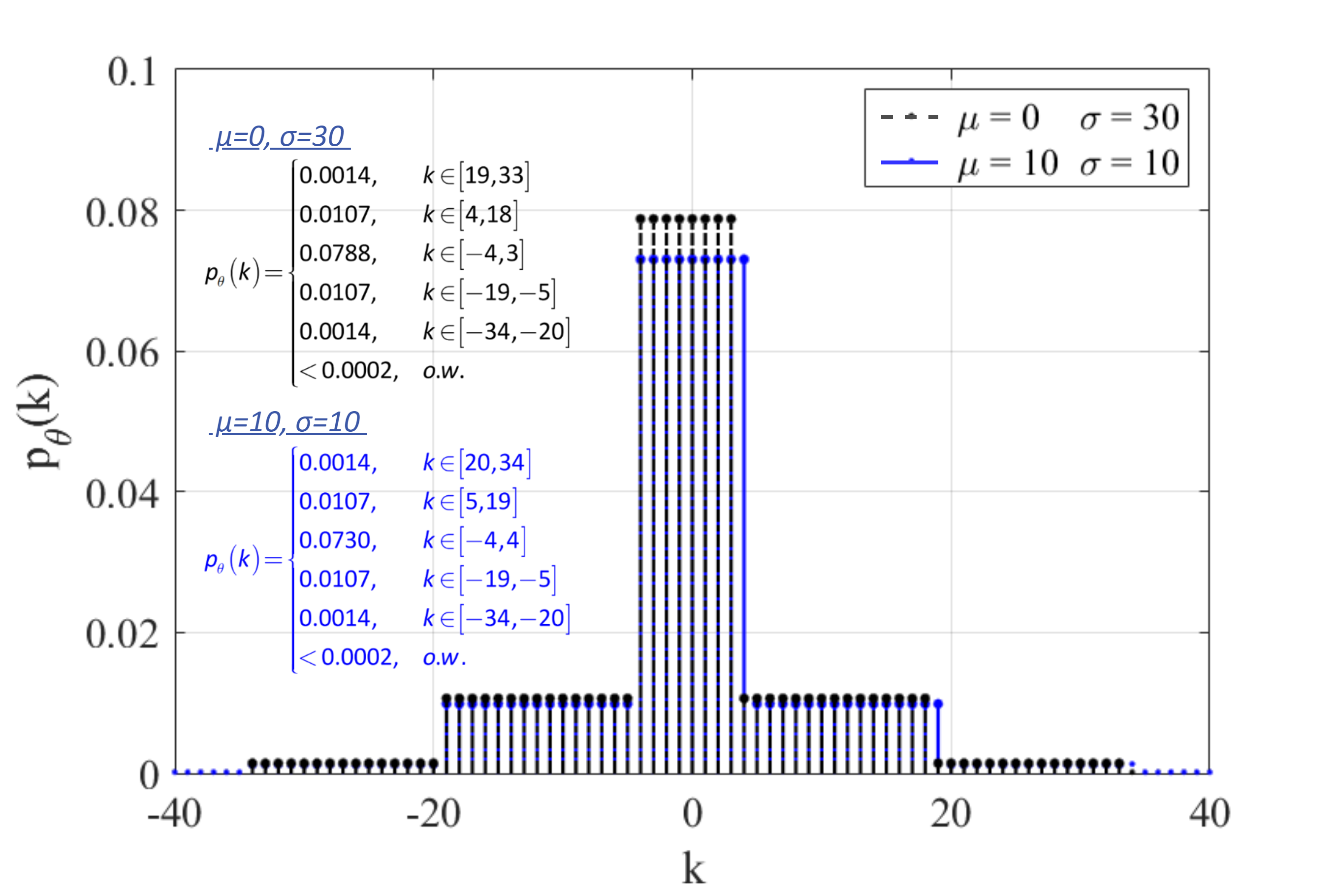}
    }
    \subfigure[Optimal distribution for the Poisson input.]{\label{fig:opti Poisson}
        \includegraphics[width=0.45\textwidth]{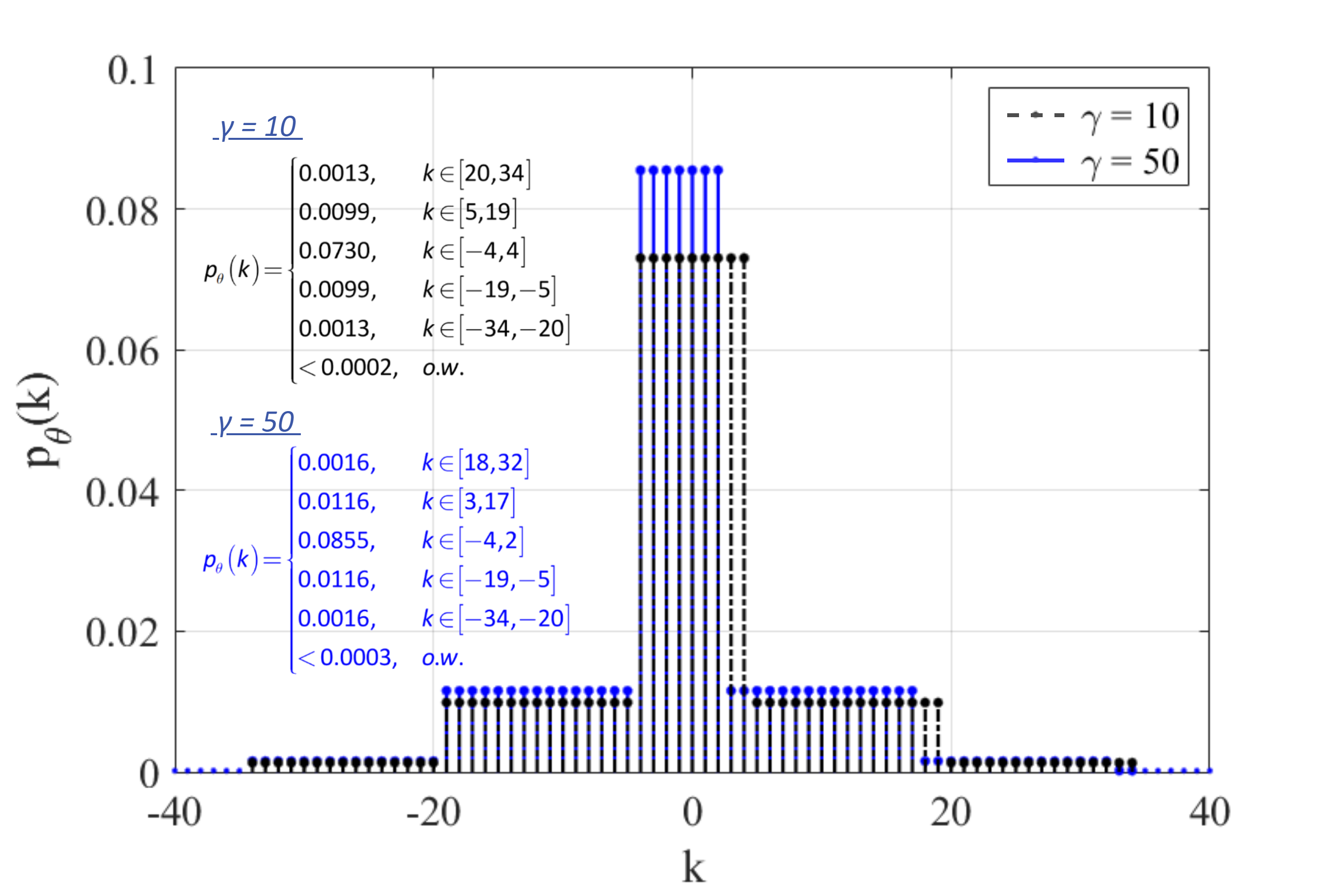}
    }
    \caption{Effect of input distributions on the optimal mechanism.}
    \label{fig:Performance comparison 01}
\end{figure}
Given that the Gaussian and the Poisson distribution are two commonly used distributions \cite{roy2003discrete, shmueli2005useful}, we select these two as the input instances.
The discrete Gaussian PMF is a discretization result, with the mean parameter $\mu$, and the variance $\sigma^2$:
\begin{align*}
    {p_x}\left( k \right) = \int_k^{k + 1} {\frac{1}{{\sqrt {2\pi } \sigma }}{e^{ - \frac{{{{\left( {x - \mu } \right)}^2}}}{{2{\sigma ^2}}}}}dx},~ k\in \mathcal{Z}_{\Delta=1}.
\end{align*}
and the discrete Poisson PMF with parameter $\gamma$ is given by:
\begin{align*}
    {p_x}\left( k \right) = \frac{\gamma ^k}{k!}{e^{ - \gamma }},~ k\in \mathcal{Z}_{\Delta=1}.
\end{align*}
Fig. \ref{fig:Performance comparison 01} shows the optimal distributions with the discrete Gaussian input and the Poisson input, respectively, with the DP constraints set by $\epsilon=2$ and $m=15$.
Overall, the two optimal stair distributions are all Staircase-shaped, with the height and width of the stairs influenced by the input distributions, as we have expected.

\begin{itemize}
    \item The effect of the privacy cost $\epsilon$
\end{itemize}
\setcounter{subfigure}{0}
\begin{figure}
    \centering
    \subfigure[Input = Gaussian ($\mu\!=\!0, \sigma\!=\!10$).]{\label{fig:dis-epsilon Gaussian}
        \includegraphics[width=0.45\textwidth]{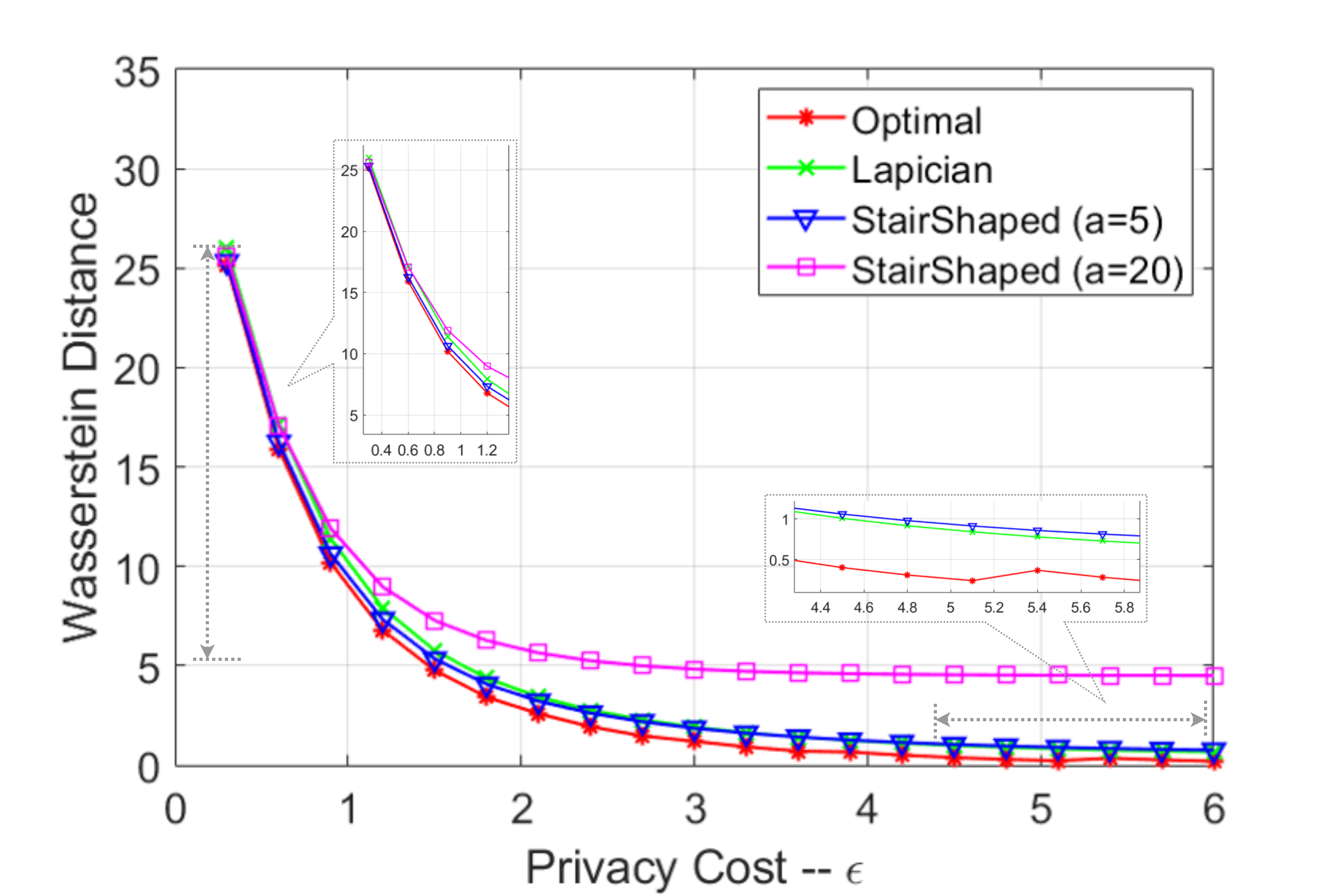}
    }
    \subfigure[Input = Poisson ($\gamma=5$).]{\label{fig:dis-epsilon Poisson}
        \includegraphics[width=0.45\textwidth]{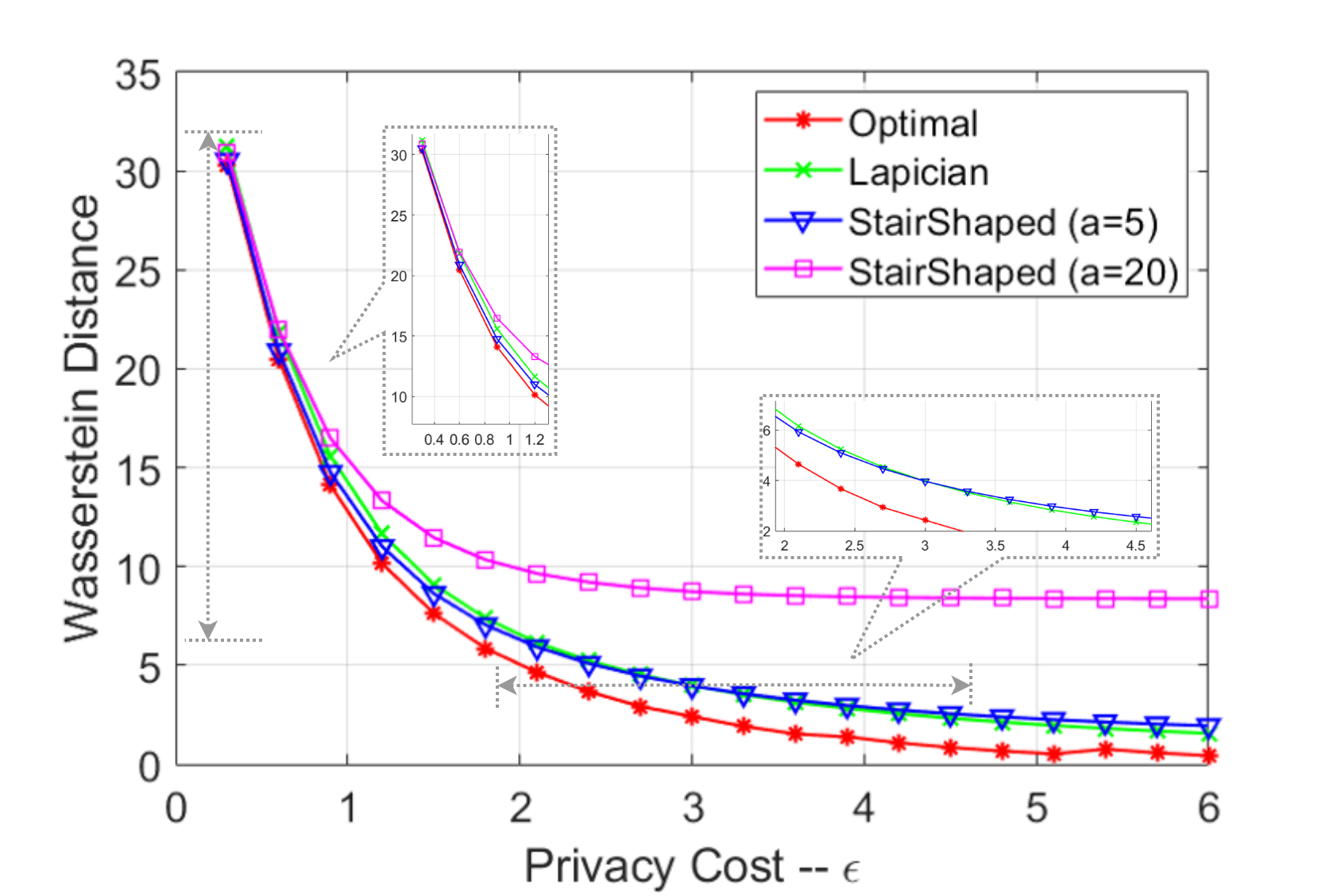}
    }
    \caption{Effect of privacy costs on the mechanism utility with different inputs.}
    \label{fig:Performance comparison 02}
\end{figure}
In this part, we assume a moderate adjacency $m=15$.
In the trade-off problem, higher privacy cost (a bigger $\epsilon$) implies less utility loss, as reflected in smaller Wasserstein distance.
This trend can be verified in Fig. \ref{fig:Performance comparison 02}, with the Gaussian and the Poisson distributed inputs, respectively.
Both two figures compare the utility performance of four $\epsilon$-DP mechanisms.
Notice that the result with the Laplacian noise (green cross line) and the Staircase-shaped noise with smaller stair width (blue triangle line) are comparable, especially in the high privacy regime (smaller $\epsilon$).
This is because the two distributions are similar under the parameter settings subject to the same DP constraints.
Further, as $\epsilon$ increases, the optimality of our mechanism is better represented, with the Wasserstein distance approaching zero, i.e., the statistical properties of the mechanism's input and output can be retained.

\begin{itemize}
    \item The effect of the adjacency $m$
\end{itemize}

With the other DP parameter set as $\epsilon=2$, the effect of the adjacency on the Wasserstein distance is shown in Fig. \ref{fig:Performance comparison 03}.
We also select four mechanisms for comparison.
It is observed that the Wasserstein distance and the adjacency are positively correlated.
If the elements in two datasets differ significantly, then even with the optimal $\epsilon$-DP mechanism, the utility guarantee is limited.
Since the larger differences require the noises of greater amplitudes, the mechanism utility is significantly sacrificed.
\setcounter{subfigure}{0}
\begin{figure}[ht]
    \centering
    \subfigure[Input = Gaussian ($\mu\!=\!0, \sigma\!=\!10$).]{\label{fig:dis-m Gaussian}
        \includegraphics[width=0.45\textwidth]{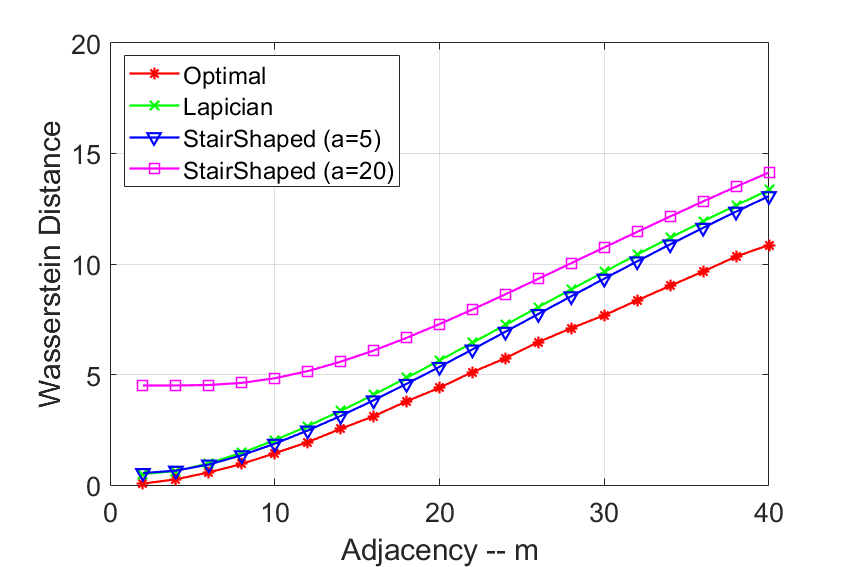}
    }
    \subfigure[Input = Poisson ($\gamma=5$).]{\label{fig:dis-m Poisson}
        \includegraphics[width=0.45\textwidth]{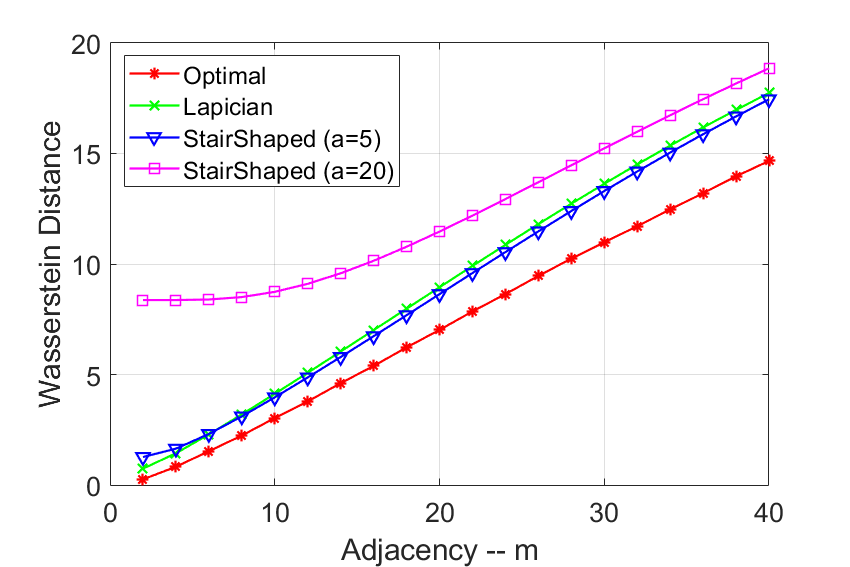}
    }   
    \caption{Effect of adjacencies on the mechanism utility with different inputs.}
    \label{fig:Performance comparison 03}
\end{figure}
\begin{remark}
    In the above simulation, we compare the mechanism utility with the optimal Staircase-shaped distribution with unfixed parameters (red circle line), and two standard Staircase-shaped distributions with explicit parameters (blue triangle line and purple square line).
    Although they all satisfy the $\epsilon$-DP constraints, the standard fixed \textit{Staircase} mechanism performs slightly worse in guaranteeing the utility.
    Since the utility measure we define is related to the mechanism input, we couple the optimal noise probability distribution with the input.
    The effectiveness is confirmed by the simulation.
    Due to the arbitrariness of the inputs, we are unable to give a closed-form expression of the optimal noise distribution independent of the input.
    Instead, we obtain the optimal mechanism with unfixed parameters, by solving an equivalent problem through linear programming.
\end{remark}

\subsection{Verification of the Mechanism Optimality}
To further validate the mechanism optimality, we compare the statistical properties of the mechanism utility with three other mechanisms, under the same $\epsilon$-DP guarantees.
To eliminate the uncertainty of the discrete random noise, we conduct 100 simulation runs for each simulation, and do frequency statistics on the mechanism utility (characterized by the Wasserstein distance).
Notice that the smaller Wasserstein distance implies the higher mechanism utility.

\begin{figure}[ht]
    \begin{center}
    \includegraphics[width=0.63\textwidth]{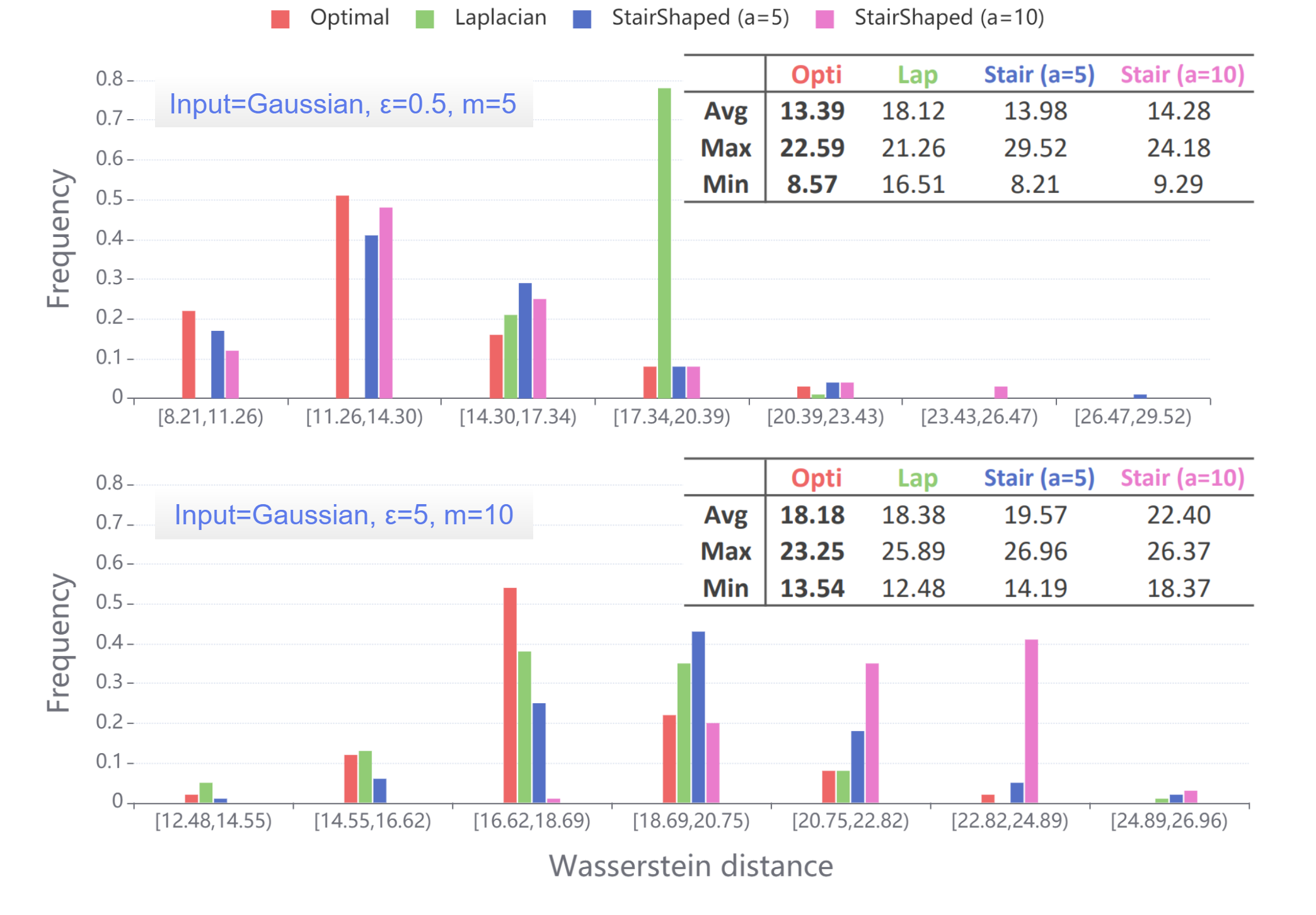}
    \vspace{-15pt}
    \caption{Statistic properties of the Wasserstein Distance.}
    \label{fig:Statistic properties of the Wasserstein Distance}
    \end{center}
\end{figure}
Fig. \ref{fig:Statistic properties of the Wasserstein Distance} shows a comparison of the utility for four $\epsilon$-DP mechanisms, with different privacy level and adjacency parameters settings.
Overall, our proposed mechanism has a higher probability of the small Wasserstein distance, i.e., the higher mechanism utility.
Further, to compare the mechanism utility more clearly, we summarize the results (the average, maximum, minimum Wasserstein distance) in Fig. \ref{fig:Statistic properties of the Wasserstein Distance}.
In some cases, the utility guaranteed by the optimal mechanism is similar to the existing mechanism.
For example, with higher privacy protection ($\epsilon=0.5$), its utility is similar to the \textit{Staircase} mechanism;
in the low privacy regime, the performance is close to the \textit{Laplacian} mechanism.
Note that the utility for the optimal mechanism is not always the highest, partially due to the uncertainty of the discrete random noise.
In conclusion, the optimal $\epsilon$-DP mechanism ensures the maximum mechanism utility in the vast majority of cases.

\section{Conclusion} \label{sec:Conclusion}
For the discrete random noise adding mechanisms, we considered the DP conditions, properties and the trade-off between the mechanism utility and privacy level.
For the general DP mechanisms, a sufficient and necessary condition for $\epsilon$-DP and a sufficient condition for $(\epsilon,\delta)$-DP were derived, followed by the DP parameters estimation.
Afterwards, based on the conditions, we analyzed the DP properties for several typical mechanisms.
Furthermore, we took the Wasserstein distance between mechanism inputs and outputs as the utility metric, and built the trade-off issue as a utility-maximization optimization problem.
The proposed optimal mechanism is \textit{Staircase}-shaped, with the parameters depending on the mechanism inputs and the differential privacy requirements.
Extensive simulations were performed to verify its optimality.
Future directions include the DP conclusions extension for more individuals, and exploring the correlation between the differential privacy with the homomorphic encryption.


\appendices
\section{Proof of Theorem \ref{theorem-epsilon}} \label{proof 01}
\begin{proof}
$\Leftarrow$: We prove the necessity by contradiction.
    Assume that
    \begin{equation*}
        \underset{\forall m_0\in \left[ -m,m \right],m_0\in \mathcal{Z}_{\Delta},\forall i\in V}{\mathop{\sup }}\,\frac{{p_{{\theta_i} }}\left( k-m_0 \right)}{{p_{{\theta_i}}}\left( k \right)}=\infty,~
        k\in \mathcal{Z}_{\Delta},
    \end{equation*}
    i.e., for any given large constant $M$, there exists $k_0 \in \mathcal{Z}_{\Delta}$ such that
    \begin{align*}
        \frac{{p_{\theta_i}}\left( k_0- m_0  \right)}{{p_{{\theta_i}}}\left( k_0 \right)}\geq M,
    \end{align*}
    where $m_0\in \mathcal{Z}_{\Delta}\setminus\{0\}$.
    Construct a pair of $m_0$-adjacent state vectors $x_i,y_i \in \mathcal{Z}_{\Delta}$ satisfying
    \begin{equation*}
        \forall i\in V,~
        {\left| {{x_i} - {y_i}} \right|} \le \left\{ {\begin{aligned}
        &m_0, &&{i = {i_0}}; \\
        &0, &&{i \ne {i_0}}.
        \end{aligned}} \right.
    \end{equation*}
    Based on the discrete property of $\mathcal{Z}_\Delta$ and the sign of $m_0$, we divide $m_0 \in \left[-m, m\right]$ into three parts: $m_0 \in \left\{ -\Delta, \Delta \right\}$, $\left\{ j|2\Delta \le j\le m,j\in \mathcal{Z}_{\Delta} \right\}$ and $\left\{ j|-m\le j\le -2\Delta,j\in \mathcal{Z}_{\Delta} \right\}$. 
    Denote $\mathcal{O} \subseteq \mathcal{S}$, where $\mathcal {O}_i$ is a set of the $i$-th column element in $\mathcal{O}$.
    Note that DP is guaranteed if (\ref{eq:privacy}) holds for any given $\mathcal{O}$.
    In the following three parts, we construct the output range $\mathcal{O}_{i_0}$ respectively to derive the contradiction for the necessity proof.
    
    \begin{itemize}
        \item $m_0\in \left\{ -\Delta, \Delta \right\}$
    \end{itemize}
    Define ${{\mathcal{O}}_{i_0}}=\left\{ k|k=y_{i_0}+k_0 \right\}$.
    From (\ref{eq:mechanism}), we have
    \begin{align} \label{11}
        &\frac{\Pr \left\{ \mathcal{A}\left( {{x}_{i_0}} \right)\in {{\mathcal{O}}_{i_0}} \right\}}{\Pr \left\{ \mathcal{A}\left( y_{i_0} \right)\in {{\mathcal{O}}_{i_0}} \right\}}
        =\frac{{p_{{{x}_{i_0}}+{\theta_{i_0}}}}\left( k \right)}{{p_{y_{i_0}+{\theta_{i_0}}}}\left( k \right)}\nonumber\\
        = &\frac{{p_{{{x}_{i_0}}+{\theta_{i_0}}}}\left( y_{i_0}+k_0 \right)}{{p_{y_{i_0}+{\theta_{i_0}}}}\left( y_{i_0}+k_0 \right)}
        =\frac{{p_{{\theta_i}}}\left( k_0- m_0  \right)}{{p_{{\theta_i}}}\left( k_0 \right)}
        = M.
    \end{align}
    
    \begin{itemize}
        \item $m_0\in \left\{ j|2\Delta \le j\le m,j\in \mathcal{Z}_{\Delta} \right\}$
    \end{itemize}
    Define ${{\mathcal{O}}_{i_0}}=\left\{ k|y_{i_0}+k_0\le k\le y_{i_0}+k_0+ m_0-\Delta \right\}$.
    Since $p_{{\theta_{i_0}}}$ is bounded, there exists a constant $C\ge 1$, s.t., 
    \begin{align*}
        \sum\nolimits_{k=k_0}^{k_0+m_0-\Delta}{{p_{{\theta_{i_0}}}}\left( k \right)}
        & ={p_{{\theta_{i_0}}}}\left( k_0 \right)
            +\sum\nolimits_{k=k_0+\Delta }^{k_0+m_0-\Delta}{{p_{{\theta_{i_0}}}}\left( k \right)} \nonumber\\
        & \le C\cdot {p_{{\theta_{i_0}}}}\left( k_0 \right).
    \end{align*}
    Then, one follows that 
    \begin{align} \label{22}
        & \frac{\Pr \left\{ \mathcal{A}\left( {{x}_{i_0}} \right)\in {{\mathcal{O}}_{i_0}} \right\}}{\Pr \left\{ \mathcal{A}\left( y_{i_0} \right)\in {{\mathcal{O}}_{i_0}} \right\}} \nonumber\\ 
        =& \frac
            {\sum\nolimits_{k=y_{i_0}+k_0}^{y_{i_0}+k_0+m_0-\Delta }{{p_{{{x}_{i_0}}+{\theta_{i_0}}}}\left( k \right)}}
            {\sum\nolimits_{{k=y_{i_0}}+k_0}^{y_{i_0}+k_0+m_0-\Delta }{{p_{y_{i_0}+{\theta_{i_0}}}}\left( k \right)}}
        =\frac
            {\sum\nolimits_{k=k_0-m_0}^{k_0-\Delta }{p_{\theta_{i_0}}\left( k \right)}}
            {\sum\nolimits_{k=k_0}^{k_0+m_0-\Delta }{p_{\theta_{i_0}}\left( k \right)}} \nonumber\\ 
        =& \frac
            {{p_{{\theta_{i_0}}}}\left( k_0- m_0  \right)+\sum\nolimits_{k=k_0-m_0+\Delta}^{k_0-\Delta }{p_{\theta_{i_0}}\left( k \right)}}
            {\sum\nolimits_{k=k_0}^{k_0+m_0-\Delta }{p_{\theta_{i_0}}\left( k \right)}} \nonumber\\
        \ge& \frac
            {{p_{{\theta_{i_0}}}}\left( k_0- m_0  \right)}
            {\sum\nolimits_{k=k_0}^{k_0+m_0-\Delta }{p_{\theta_{i_0}}\left( k \right)}}
        \ge \frac
            {{p_{\theta_{i_0}}}\left( k_0- m_0 \right)}
            {C\cdot {p_{{\theta_{i_0}}}}\left( k_0 \right)}
        =\frac{M}{C}.
    \end{align}
    
    \begin{itemize}
        \item $m_0\in \left\{ j|-m\le j\le -2\Delta,j\in \mathcal{Z}_{\Delta} \right\}$
    \end{itemize}
    Define ${{\mathcal{O}}_{i_0}}=\left\{ k|y_{i_0}+k_0+m_0+\Delta \le k\le y_{i_0}+k_0 \right\}$.
    Similarly, we have
    \begin{align}\label{33}
        \!\!\!\! & \frac{\Pr \left\{ \mathcal{A}\left( {{x}_{i_0}} \right)\in {{\mathcal{O}}_{i_0}} \right\}}{\Pr \left\{ \mathcal{A}\left( y_{i_0} \right)\in {{\mathcal{O}}_{i_0}} \right\}} \nonumber\\ 
        =&\frac
            {\sum\nolimits_{k=y_{i_0}+k_0+m_0+\Delta }^{y_{i_0}+k_0}{{p_{{{x}_{i_0}}+{\theta_{i_0}}}}\left( k \right)}}
            {\sum\nolimits_{k=y_{i_0}+k_0+m_0+\Delta }^{y_{i_0}+k_0}{{p_{y_{i_0}+{\theta_{i_0}}}}\left( k \right)}} 
        =\frac
            {\sum\nolimits_{k=k_0+\Delta }^{k_0-m_0 }{p_{\theta_{i_0}}\left( k \right)}}
            {\sum\nolimits_{k=k_0+m_0+\Delta }^{k_0}{p_{\theta_{i_0}}\left( k \right)}}\!\!\! \nonumber\\ 
        =&\frac
            {p_{\theta_{i_0}}\left( k_0-m_0  \right)+\sum\nolimits_{k=k_0+\Delta}^{k_0-m_0-\Delta }{p_{\theta_{i_0}}\left( k \right)}}
            {\sum\nolimits_{k=k_0+m_0+\Delta }^{k_0}{p_{\theta_{i_0}}\left( k \right)}} \nonumber\\ 
        \ge& \frac
            {p_{\theta_{i_0}}\left( k_0-m_0 \right)}
            {C\cdot {p_{\theta_{i_0}}}\left( k_0 \right)}
        =\frac{M}{C}.
    \end{align}
    
    Note that $M$ were to take any value and (\ref{11}), (\ref{22}), (\ref{33}) would violate the $(\epsilon,\delta)$-DP definition ($\delta=0$) in (\ref{eq:privacy}).
    Thus, through the contradictions, we prove that (\ref{eq:c2 discrete}) is a necessary condition for the $\epsilon$-DP mechanism $\mathcal A$.

$\Rightarrow$:  Next, we prove the sufficiency.
    Based on (\ref{eq:mechanism}), we have
    \begin{align} \label{AX}
        & \Pr \{\mathcal{A}(x)\in \mathcal{O}\}
        =\Pr \left\{ x+\theta \in \mathcal{O} \right\} \nonumber\\
        =& \Pr \left\{ {{x}_{i_0}}+{\theta_{i_0}}\in {{\mathcal{O}}_{i_0}} \right\}\prod\limits_{i=1,i\ne i_0}^{n}{\Pr }\left\{ {{x}_i}+{\theta_i}\in {{\mathcal{O}}_i} \right\}
    \end{align}
    and
    \begin{align} \label{AY}
        & \Pr \{\mathcal{A}(y)\in \mathcal{O}\}
        =\Pr \left\{ y+\theta \in \mathcal{O} \right\} \nonumber\\
        =& \Pr \left\{ {y_{i_0}}+{\theta_{i_0}}\in {{\mathcal{O}}_{i_0}} \right\}\prod\limits_{i=1, i\ne i_0}^{n}{\Pr }\left\{ {y_i}+{\theta_i}\in {{\mathcal{O}}_i} \right\}.
    \end{align}
    Due to $x_i=y_i, i\neq i_0$, we have
    \begin{align} \label{XY}
        \prod\limits_{i=1,i\ne i_0}^{n} \!\!\!\!\!  {\Pr }\left\{ {{x}_i}+{\theta_i}\in {{\mathcal{O}}_i} \right\}=\!\!\!\prod\limits_{i=1,i\ne i_0}^{n} \!\!\!\!\! {\Pr }\left\{ {y_i}+{\theta_i}\in {{\mathcal{O}}_i} \right\}.
    \end{align}
    Besides, with the condition in (\ref{eq:c2 discrete}),  it follows that
    \begin{align} \label{x<y}
     & \Pr \left\{ {{x}_{i_0}}+{\theta_{i_0}}\in {{\mathcal{O}}_{i_0}} \right\} 
     =\sum\nolimits_{{{\mathcal{O}}_{i_0}}}{{p_{{{x}_{i_0}}+{\theta_{i_0}}}}\left( k \right)} \nonumber\\ 
     =&  \sum\nolimits_{{{\mathcal{O}}_{i_0}}}{{p_{y_{i_0}+ m_0 +{\theta_{i_0}}}}\left( k \right)} 
     \le \sum\nolimits_{{{\mathcal{O}}_{i_0}}}{{c_b} \cdot {p_{{y_{i_0}}+{\theta_{i_0}}}}\left( k \right)} \nonumber\\ 
     =& {c_b}\Pr \left\{ {y_{i_0}}+{\theta_{i_0}}\in {{\mathcal{O}}_{i_0}} \right\}.
    \end{align}
    Combining (\ref{AX})-(\ref{x<y}), it yields that
    \begin{align} \label{proof 1}
        \Pr \{\mathcal{A}(x)\in \mathcal{O}\}
        & \le {c_b}\Pr \{\mathcal{A}(y)\in \mathcal{O}\} \nonumber\\
        & ={{e}^{\log \left( {c_b} \right)}}\Pr \{\mathcal{A}(y)\in \mathcal{O}\},
    \end{align}
    which satisfies the definition of $\epsilon$-DP. 
    
    Furthermore, comparing (\ref{eq:privacy}) and (\ref{proof 1}), we can easily obtain the estimation of the DP parameter $\epsilon$, i.e.,$\epsilon = \log(c_b)$.
    From (\ref{eq:c2 discrete}), we note that the upper bound $c_b$ relies on the adjacency $m_0$.
    When $m_1 \leq m_2$ holds, we have
    \begin{align*}
        \underset{\forall m_0\in \left[ -m_1,m_1 \right]}{\mathop{\sup }}\,\frac{{p_{{\theta_i}+m_0}}\left( k \right)}{{p_{{\theta_i}}}\left( k \right)}
        \leq \underset{\forall m_0\in \left[ -m_2,m_2 \right]}{\mathop{\sup }}\,\frac{{p_{{\theta_i}+m_0}}\left( k \right)}{{p_{{\theta_i}}}\left( k \right)},
    \end{align*}
    where $m_0,m_1,m_2,k\in \mathcal{Z}_{\Delta}$. Hence, we refer that $c_b$ is an increasing function of $m$.

\end{proof}

\section{Proof of Theorem \ref{theorem-epsilon-delta-02}} \label{proof 02}
\begin{proof}
    Construct a pair of $m$-adjacent state vectors $x$ and $y$ with ${{x}_{i_0}}={y_{i_0}}+m$ and $x_i=y_i$ (when $i\neq i_0$), where $x_i,y_i \in \mathcal{Z}_{\Delta}$.
    Then, we obtain the following result:
    \begin{align} \label{proof 2}
        & \Pr \{\mathcal{A}(x)\in \mathcal{O}\} 
        =\prod\limits_{i=1}^{n}{\Pr \left\{ \mathcal{A}({x_i})\in {{\mathcal{O}}_i} \right\}}\nonumber\\ 
        =& \Pr \left\{ \mathcal{A}({{x}_{i_0}})\in {{\mathcal{O}}_{i_0}} \right\}\prod\limits_{i=1,i\ne i_0}^{n}{\Pr }\left\{ \mathcal{A}({{x}_i})\in {{\mathcal{O}}_i} \right\} \nonumber\\
        =& \bigg[ \Pr \left\{ \mathcal{A}({{x}_{i_0}})\in {{\mathcal{O}}_{i_0}}\left| \theta \in {\Theta_1} \right. \right\} + \Pr \left\{ \mathcal{A}({{x}_{i_0}})\in {{\mathcal{O}}_{i_0}}\left| \theta \in {\Theta_0} \right. \right\} \bigg] \nonumber\\
            & \quad\times \prod\limits_{i=1,i\ne i_0}^{n}{\Pr }\left\{ \mathcal{A}({{x}_i})\in {{\mathcal{O}}_i} \right\} \nonumber\\
        =& \left[ \sum\nolimits_{{{\mathcal{O}}_{i_0}},\theta \in {\Theta_1}}{{p_{{{x}_{i_0}}+{\theta_{i_0}}}}\left( k \right)}+\sum\nolimits_{{{\mathcal{O}}_{i_0}},\theta \in {\Theta_0}}{{p_{{{x}_{i_0}}+{\theta_{i_0}}}}\left( k \right)} \right] \nonumber\\
            & \quad\times \prod\limits_{i=1,i\ne i_0}^{n}{\Pr }\left\{ \mathcal{A}({x_i})\in {{\mathcal{O}}_i} \right\} \nonumber\\ 
        =& \left[ \sum\nolimits_{{{\mathcal{O}}_{i_0}},\theta \in {\Theta_1}}{{p_{{y_{i_0}}+m+{\theta_{i_0}}}}\left( k \right)}+\sum\nolimits_{{{\mathcal{O}}_{i_0}},\theta \in {\Theta_0}}{{p_{{{x}_{i_0}}+{\theta_{i_0}}}}\left( k \right)} \right] \nonumber\\
            & \quad\times \prod\limits_{i=1,i\ne i_0}^{n}{\Pr }\left\{ \mathcal{A}({y_i})\in {{\mathcal{O}}_i} \right\} \nonumber\\
        \le& {c_b}\sum\nolimits_{{{\mathcal{O}}_{i_0}},\theta \in {\Theta_1}}{{p_{y_{i_0}+ {\theta_{i_0}} }}\left( k \right)}\prod\limits_{i=1,i\ne i_0}^{n}{\Pr }\left\{ \mathcal{A}({y_i})\in {{\mathcal{O}}_i} \right\} \nonumber\\ 
            & \quad+\sum\nolimits_{{\Theta_0}}{{p_{{\theta_{i_0}}}}\left( k \right)}\prod\limits_{i=1,i\ne i_0}^{n}{\Pr }\left\{ \mathcal{A}({y_i})\in {{\mathcal{O}}_i} \right\} \nonumber\\ 
        =& {c_b}\Pr \left\{ \mathcal{A}({y_{i_0}})\in {{\mathcal{O}}_{i_0}} \right\}\prod\limits_{i=1,i\ne i_0}^{n}{\Pr }\left\{ \mathcal{A}({y_i})\in {{\mathcal{O}}_i} \right\} \nonumber\\ 
            & \quad+\sum\nolimits_{{\Theta_0}}{{p_{{\theta_{i_0}}}}\left( k \right)}\prod\limits_{i=1,i\ne i_0}^{n}{\Pr }\left\{ \mathcal{A}({y_i})\in {{\mathcal{O}}_i} \right\} \nonumber\\ 
        \le& {c_b}\Pr \{\mathcal{A}(y)\in \mathcal{O}\}+\delta.
    \end{align}
    Combining (\ref{eq:privacy}) and (\ref{proof 2}), we conclude that
    \begin{align*}
        c_b=e^\epsilon \Rightarrow \epsilon = \log(c_b).
    \end{align*}
    Thus, we have completed the proof.
\end{proof}

\section{Proof of Theorem \ref{th:Gaussian}} \label{proof 03}
\begin{proof}
    First, we prove that the \textit{Gaussian} mechanism is not $\epsilon$-DP.
    Due to the symmetry of the term $m_0$, we take $m_0\in \left\{ j|0 \le j\le m,j\in \mathcal{Z}_{\Delta} \right\}$ in this proof.
    The following proof can be applied to the situation $m_0<0$ similarly.
    Based on the relationship between the probability point $k$ and the PMF parameter $\mu$, we divide the point $k$ into three parts.
    
    \textit{Case 1:} $k \ge \mu+m_0$.
    Here the PMF is decreasing.
    Based on the PMF in (\ref{eq:discrete Gaussian distribution}), we obtain that
    \begin{align*}
        & \left| \frac
            {p_{\theta_i}\left( k-m_0 \right)}
            {p_{\theta_i}\left( k \right)} \right|
        = \frac
            {{\int_{{k_0}}^{k + \Delta } {{e^{ - \frac{{{{\left( {{\theta _i} - m_0- \mu } \right)}^2}}}{{2{\sigma ^2}}}}}d{\theta _i}} }}
            {{\int_k^{k + \Delta } {{e^{ - \frac{{{{\left( {{\theta _i} - \mu } \right)}^2}}}{{2{\sigma ^2}}}}}d{\theta _i}} }} \\\nonumber
        \le& \frac
            {{\Delta \!\cdot\! {e^{ - \frac{1}{{2{\sigma ^2}}}{{\left( {k - m_0- \mu } \right)}^2}}}}}
            {{\Delta \!\cdot\! {e^{ - \frac{1}{{2{\sigma ^2}}}{{\left( {k + \Delta  - \mu } \right)}^2}}}}} 
        \!=\! {e^{\frac{1}{{2{\sigma ^2}}}\left( {m_0+ \Delta } \right)\left( {2k + \Delta - m_0 - 2\mu } \right)}}
        \!=\! Q. \!\!
    \end{align*}
    
    \textit{Case 2:} $k \le \mu - \Delta$.
    Here the PMF is increasing.
    Similarly, we have
    \begin{align*}
        \left| \frac
            {p_{\theta_i}\left( k-m_0 \right)}
            {p_{\theta_i}\left( k \right)} \right|
        \le& \frac{{\Delta  \cdot {e^{ - \frac{1}{{2{\sigma ^2}}}{{\left( {k + \Delta -m_0- \mu } \right)}^2}}}}}{{\Delta  \cdot {e^{ - \frac{1}{2\sigma ^2}{\left( {k - \mu } \right)}^2}}}} \\\nonumber
        =& {e^{\frac{1}{{2{\sigma ^2}}}\left( {{m_0} - \Delta } \right)\left( {2k + \Delta -m_0- 2\mu } \right)}}
        =R.
    \end{align*}
        
    \textit{Case 3:} $ \mu \le k \le \mu+m_0-\Delta$.
    It shows that
    \begin{align*}
        \left| \frac
            {p_{\theta_i}\left( k-m_0 \right)}
            {p_{\theta_i}\left( k \right)} \right|
        \le& \frac{{\Delta  \cdot {e^{ - \frac{1}{{2{\sigma ^2}}}{{\left( {k + \Delta -m_0- \mu } \right)}^2}}}}}{{\Delta  \cdot {e^{ - \frac{1}{{2{\sigma ^2}}}{{\left( {k + \Delta  - \mu } \right)}^2}}}}} \\\nonumber
        =&  {e^{\frac{1}{{2{\sigma ^2}}}{m_0}\left( {2k + 2\Delta -m_0- 2\mu } \right)}}
        =S.
    \end{align*}
    It shows that with the upper bounds of $k$ in \textit{Case 1,2}, the upper bounds $R$ and $S$ exist.
    However, we have $Q\to \infty$ as $k$ increases.
    According to Theorem \ref{theorem-epsilon}, there does not exist a bounded parameter $c_b$ to guarantee the finite privacy loss.
    Thus, we have that the \textit{Gaussian} mechanism is not $\epsilon$-DP.

    Next, we apply Theorem \ref{theorem-epsilon-delta-02} to prove that the \textit{Gaussian} mechanism guarantees $(\epsilon,\delta)$-DP.
    Given a constant $M \ge m$, for $k \in \left[{ - M,M} \right]$, the DP parameter $\epsilon$ is bounded by
    \begin{align} \label{Gaussian epsilon}
        \epsilon  
        &= \log(c_b)
            \le \log \left( {\max \left\{ {Q,R,S} \right\}} \right) \\ \nonumber
        &\le \mathop {\max }\limits_{m_0 \in \left[ {- m,m} \right]}\! \left[ {\frac1{{2{\sigma ^2}}}\left( {m_0\!+\!\Delta} \right)\left( {2M\! +\! \Delta\! -\! {m_0}\! -\! 2\mu } \right)} \right]\!\!\!\! \\ \nonumber 
        &= \frac{1}{{2{\sigma ^2}}}\left( {m + \Delta } \right)\left( {2M + \Delta  - m - 2\mu } \right).
    \end{align}
    Meanwhile, the probability of error $\delta$ is bounded by
    \begin{align} \label{Gaussian delta}
        \delta  
        & \le \sum\nolimits_{\theta_i \in {\Phi _i}} {{p_{{\theta _i}}}\left( {k } \right)} \nonumber\\
        & = \frac1{\sqrt {2\pi \sigma^2 }}\sum\nolimits_{\theta_i \in {\Phi _i}} {\int_{k }^{k + \Delta} {e^{ - \frac{{\left( {\theta_i - \mu} \right)}^2}{2{\sigma^2}}}d\theta_i} },
    \end{align}
    where $\Phi _i = \left( { - \infty , - M} \right] \cup \left[ {M,\infty } \right) \subseteq \mathcal{Z}_\Delta$ is the $i$-th dimensional noise range of $\mathcal{O}_i$. 
    Finally, one infers that the \textit{Gaussian} mechanism is $({\epsilon,\delta})$-DP, where the DP parameters $\epsilon$ and $\delta$ are estimated by (\ref{Gaussian epsilon}) and (\ref{Gaussian delta}), respectively. 
    
\end{proof}

\section{Proof of Theorem \ref{theorem-LP}} \label{proof LP}
\begin{proof}
    The key idea of the equivalence proof is to make the optimization variables $p_\theta$ explicitly involved in the problem, which is a basis to convert the primal problem into a conventional convex optimization problem.
    To make the derivation more clear, we represent the mechanism input probability distribution as a finite one:
    \begin{align} \label{x-simplified}
        \!\! p_x \!=\!
        \left(\!\! {\begin{array}{*{20}{c}}
            p_x(k_0\!+\!\Delta),\!\!\!\!\! &p_x(k_0\!+\!2\Delta),\!\!\!\!\! &\dots,\!\!\!\!\! &p_x(k_0\!+\!n\Delta)\!\!\!
        \end{array}} \right)^T \!\!,
    \end{align}
    where $k_0\in \mathcal{Z}_\Delta$ and $n\in \mathbb{R}$.
    Note that $k_0$ is an arbitrary value and $n\rightarrow \infty$.
    We can have the finite input (\ref{x-simplified}) replace the infinite one (\ref{new p_x}).
    Our goal is to prove that $W_0=W_1$, where
    \begin{align*}
        W_0 = \sum\nolimits_k {\left| {{P_x}\left( k \right) - {P_{x + \theta }}\left( k \right)} \right|},~
        W_1 = \sum\nolimits_k {\left| {{p_x^T} M_k {p_\theta}} \right|}.
    \end{align*}
    
    First, we discuss the equivalence of \textit{the objective variables}.
    With the introduction of the noise distribution $p_\theta$ in (\ref{new p_theta}), the original objective variable $p_\theta(k)$ aiming at every probability value is contained in this column vector $p_\theta$, which is expressed equivalently but more concisely.
    
    Next, we prove the equivalence of \textit{two objective functions}, $W_0$ and $W_1$.
    We aim to convert the CDF into PMF, which contains the optimization variables more explicitly.
    \begin{itemize}
        \item The input cumulative probability at point $k$: ${P_x}\left( k \right)$.
    \end{itemize}
    \vspace{-2pt}
    \begin{align} \label{proof:seg1}
        {P_x} \left( k \right)
         = p_x(-\infty) + \cdots + p_x(k-\Delta) + p_x(k). 
    \end{align}
    Combining (\ref{x-simplified}) and (\ref{proof:seg1}), we easily have
    \begin{align*}
        {P_x}\left( k \right) =\left\{ \begin{aligned}
            &0, && k\le k_0; \\
            &1, && k \ge k_0+n\Delta,  \\
        \end{aligned} \right.
    \end{align*}
    due to the boundedness property of the CDF:
    \begin{align*}
        \mathop {\lim }\limits_{k \to -\infty } {P_x}\left( k \right) = 0,\mathop {\lim }\limits_{k \to +\infty } {P_x}\left( k \right) = 1.
    \end{align*}
    Then, for $k \in \left\{ {{k^{'}}\left| {{k_0} + \Delta  \le {k^{'}} \le {k_0} + n\Delta ,{k^{'}} \in {Z_\Delta }} \right.} \right\}$,
    \begin{align*}
        {P_x}\left( k \right)
        &= p_x({k_0}+\Delta) + p_x({k_0}+2\Delta) + \cdots + p_x(k)\\\nonumber
        &=\!\! \sum\limits_{i=k_0+\Delta }^k \!\!\!{{p_x}\left( i \right)} \!\cdot\!
            \underbrace {\left( { \cdots +  {p_\theta }\left( {k \!-\! \Delta } \right) +   {p_\theta }\left(k \right) + \cdots} \right)}_{ = 1}\\\nonumber
        &= p _x^T\underbrace {\left( {\begin{array}{*{20}{c}}
         - & 1 & - \\
           & \vdots & \\
         - & 1 & - \\
         - & 0 & - \\
           & \vdots &
        \end{array}} \right)}_{M_k^1}{p _\theta}
        = p_x^T\underbrace {\left( {\begin{array}{c}
            {\mathbf{1}_{r}}\\
            {\mathbf{0}_{n-r}}
            \end{array}} \right)}_{M_k^1}{p_\theta },
    \end{align*}
    where the row number $r$ of the matrix $M_k^1$ with all elements $1$ is related to the relationship between $k$ and $k_0$, i.e., 
    \begin{align*}
        r = 1 + {{\left[ {k - \left( {{k_0} + \Delta } \right)} \right]} \mathord{\left/ {\vphantom {{\left[ {k - \left( {{k_0} + \Delta } \right)} \right]} \Delta }} \right. \kern-\nulldelimiterspace} \Delta }.
    \end{align*}
    The column number of the matrix $M_k^1$ depends on the noise distribution $p_\theta$.
    Denote the matrix $M_k^1$ as the result derived from the input ${P_x}$ at point $k$, and the matrix $M^2$ as the one from the output ${P_{x+\theta}}$, which will be obtained later.
    Specifically, we have $M_k^1 = \mathbf{0}$ for $k<k_0+\Delta$ and $M_k^1 = \mathbf{1}$ for $k>k_0+n\Delta$, where $\mathbf{0}$ and $\mathbf{1}$ are two matrixes full of elements $0$ and $1$, respectively.
    
    \begin{itemize}
        \item The output cumulative probability at point $k$: ${P_{x+\theta}}\!\left( k \right)$.
    \end{itemize}
    \vspace{-2pt}
    \begin{align} \label{proof:seg2}
        {P_{x+\theta}}\!\left( k \right)
        \!= p_{x+\theta}(\!-\!\infty\!)\! +\! \cdots + p_{x+\theta\!}(k-\Delta) + p_{x+\theta}(k).\!\!
    \end{align}
    With the finite representation of the input distribution $p_x$ in (\ref{x-simplified}), the output probability at point $k$ in (\ref{convolution}) is reformulated as following, where the upper and the lower bounds of the summation is further clarified:
    \begin{align} \label{proof:convolution}
        p_{x + \theta}\left(k \right) = \sum\nolimits_{i = {k_0} + \Delta }^{{k_0} + n\Delta } {{p_x}\left( i \right){p_\theta }\left( {k - i} \right)},~ k,i\in \mathcal{Z}_\Delta.
    \end{align}
    Then, we substitute every element in (\ref{proof:seg2}) with (\ref{proof:convolution}), and make simplification with the goal of $p_x$ and $p_{\theta}$.
    Consequently, we have
    \begin{align*}
        {P_{x+\theta}}\left( k \right)
        =& \sum\limits_{i =  - \infty }^{k - {k_0} - \left( {n + 1} \right)\Delta } {\left( {\sum\limits_{j = {k_0} + \Delta }^{{k_0} + n\Delta } {{p_x}\left( j \right)} } \right){p_\theta }\left( i \right)} \\\nonumber
            &+ \sum\limits_{i = k - {k_0} - n\Delta }^{k - {k_0} - \Delta } {\left( {\sum\limits_{j = {k_0} + \Delta }^{k - i} {{p_x}\left( j \right)} } \right){p_\theta }\left( i \right)} \\\nonumber
        =& p _x^T
            \underbrace {
            \bordermatrix{
            &\!\!       &\!\!\!\!       &\!\!\!\![k_1]      &\!\!\!\!   &\!\!\!\!   &\!\!\!\![k_2]  &\!\!\!\!    &\!\!\!\! \cr
            &\!\!\cdots &\!\!\!\!1      &\!\!\!\!1  &\!\!\!\!\cdots &\!\!\!\!1  &\!\!\!\!1      &\!\!\!\!0  &\!\!\!\!\cdots     \cr
            &\!\!\cdots &\!\!\!\!1      &\!\!\!\!1  &\!\!\!\!\cdots &\!\!\!\!1  &\!\!\!\!0      &\!\!\!\!0  &\!\!\!\!\cdots     \cr
            &\!\!\vdots &\!\!\!\!\vdots &\!\!\!\!\vdots     &\!\!\!\!{\mathinner{\mkern2mu\raise1pt\hbox{.}\mkern2mu
             \raise4pt\hbox{.}\mkern2mu\raise7pt\hbox{.}\mkern1mu}} &\!\!\!\!\cdots &\!\!\!\!\vdots &\!\!\!\!\vdots &\!\!\!\!\vdots \cr
            &\!\!\cdots &\!\!\!\!1      &\!\!\!\!1  &\!\!\!\!0      &\!\!\!\!\cdots &\!\!\!\!0  &\!\!\!\!0  &\!\!\!\!\cdots
            }
            }_{M_k^2}
            {p _\theta},
    \end{align*}
    where the column vector $[k_t]$ corresponds to the element at the corresponding position in the column vector $p_\theta$, i.e., $p_\theta(k_t)$.
    Note that $k_1=k-k_0-n\Delta, k_2=k-k_0-\Delta$ are the two key points.
    Before the column $[k_1]$, all elements in the matrix $M_k^2$ are $1$, and after the column $[k_2]$, all elements are $0$.
    So far, we make the two CDFs at the point $k$, ${P_x(k)}$ and ${P_{x+\theta}(k)}$, related to the input distribution $p_x$ and the noise distribution $p_{\theta}$.
    The concrete information about point $k$ is contained in the matrix $M_k$, which only depends on $k$, where
    \begin{align} \label{proof:M_k}
        M_k= M_k^1-M_k^2.
    \end{align}
    Then, we have
    \begin{align}\label{proof:objective}
        W_0
        &= \sum\nolimits_k {\left| {{P_x}\left( k \right) - {P_{x + \theta }}\left( k \right)} \right|}\\\nonumber
        &= \sum\nolimits_k {\left| {p_x^T\left( {M_k^1 - M_k^2} \right){p_\theta }} \right|} = \sum\nolimits_k {\left| {p_x^T{M_k}{p_\theta }} \right|}
        = W_1.
    \end{align}
    
    Finally, we prove \textit{the constraints} in the problem $\textbf{P}_\textbf{1}$ are equivalent to the ones in $\textbf{P}_\textbf{0}$.
    \begin{itemize}
        \item The DP constraint in (\ref{P0-b})
    \end{itemize}
    Based on the necessary DP condition (\ref{eq:nece}), we have that ${p_\theta}\left( {k} \right) \neq 0$.
    Then, we rewrite this constraint as an inequality constraint (detailed expression in (\ref{final-constraint-1})),  i.e.,
    \begin{align} \label{proof:constraint-dp-2}
        \forall k,~ m_0\in \left[ -m,m \right],~
        &\frac{{{p_\theta }\left( {k-m_0} \right)}}{{{p_\theta }\left( k \right)}} \le c_b = {e^\epsilon} \\\nonumber
        \Leftrightarrow~& {{p_\theta}\!\left({k\! -\! m_0}\right)} \!-\! {e^\epsilon}\! \cdot\! {{p_\theta}\!\left( k \right)} \le 0.
    \end{align}
    Further, we transform (\ref{proof:constraint-dp-2}) into a matrix form:
    \begin{align} \label{proof:AX=b}
        \underbrace {
        \bordermatrix{
            &\!\![k_1]  &\!\!\!\!       &\!\!\!\!               &\!\!\!\![k_2]  &                   &                   &\!\!\!\![k_3]      \cr
            &\!\!{-e^\epsilon}  &\!\!\!\!\cdots &\!\!\!\!0      &\!\!\!\!1      &\!\!\!\!0          &\!\!\!\!\cdots     &\!\!\!\!0          \cr
            &\!\!\vdots &\!\!\!\!\ddots         &\!\!\!\!\vdots &\!\!\!\!\vdots &\!\!\!\!\vdots     &\!\!\!\!\ddots     &\!\!\!\!\vdots     \cr
            &\!\!0      &\!\!\!\!\cdots &\!\!\!\!{-e^\epsilon}  &\!\!\!\!1      &\!\!\!\!0          &\!\!\!\!\cdots     &\!\!\!\!0          \cr
            &\!\!0      &\!\!\!\!\cdots &\!\!\!\!0              &\!\!\!\!1      &\!\!\!\!{-e^\epsilon} &\!\!\!\!\cdots  &\!\!\!\!0          \cr
            &\!\!\vdots &\!\!\!\!\ddots &\!\!\!\!\vdots         &\!\!\!\!\vdots &\!\!\!\!\vdots     &\!\!\!\!\ddots     &\!\!\!\!\vdots     \cr
            &\!\!0      &\!\!\!\!\cdots &\!\!\!\!0              &\!\!\!\!1      &\!\!\!\!0          &\!\!\!\!\cdots     &\!\!\!\!{-e^\epsilon}
            }
        }_{A_k}
        \cdot~ {p_\theta}
        \le
        \underbrace {\left(\! {\begin{array}{*{20}{c}}
            0\\
            0\\
            \vdots \\
            0\\
            \vdots \\
            0
        \end{array}} \!\right)}_b,
    \end{align}
    where $k_1=k-m, k_2=k, k_3=k+m$.
    The three columns correspond to the elements $p_\theta(k_1),p_\theta(k_2),p_\theta(k_3)$, respectively.
    Since for every point $k$, we have $(2m+1)$ $\epsilon$-DP constrains, the row numbers of the matrix $A_k$ is $2m+1$.
    The column number relies on the length of the noise distribution $p_\theta$.
    Go through all the points $k$, and put $A_k$ together with the corresponding element $k$, we will have the equivalent form of the DP constraint:
    \begin{align} \label{proof:constrain-dp}
        \epsilon = \log(c_b)
        \Leftrightarrow \frac{{{p_\theta }\left( {k-m_0} \right)}}{{{p_\theta }\left( k \right)}} \le {e^{\epsilon}}
        \Leftrightarrow A p_{\theta} \le b,
    \end{align}
    where
    $A\!\!=\!\!{\left(\!\! {\begin{array}{*{20}{c}} \cdots &\!\!{{A_{k = 0}}} &\!\!{{A_{k = \Delta }}} &\!\!{{A_{k = 2\Delta }}} &\!\! \cdots \end{array}} \!\!\right)^T}$ and $b=\mathbf{0}^T$.
    \begin{itemize}
        \item The total sum constraint in (\ref{P0-c})
    \end{itemize}
    \begin{align} \label{proof:constrain-sum}
        \sum\nolimits_k {{p_\theta }(k)} \!=\! 1
        &\Leftrightarrow \sum\nolimits_k {\left| {{p_\theta }(k)} \right|}  \!=\! 1\\\nonumber
        &\Leftrightarrow {\left\| {{p_\theta }} \right\|_1} = 1
        \Leftrightarrow \left| {{p_\theta }} \right| = 1.
    \end{align}
    \begin{itemize}
        \item The positive value constraint in (\ref{P0-d})
    \end{itemize}
    \begin{align} \label{proof:constrain-positive}
        \forall k,{p_\theta}\!\left( k \right) \!>\! 0
        &\Leftrightarrow {\left(\!\! {\begin{array}{*{20}{c}}
             \cdots, &\!\!\!\!{{p_\theta}\!\left( \Delta  \right) \!>\! 0}, &\!\!\!{{p_\theta}\!\left( 2\Delta  \right) \!>\! 0}, &\!\!\!\cdots 
            \end{array}} \!\!\right)^T}\!\!\! \\\nonumber
        &\Leftrightarrow {\left(\!\! {\begin{array}{*{20}{c}}
             \cdots, &\!\!\!{{p_\theta }\left( \Delta  \right)}, &\!\!\!{{p_\theta }\left(2\Delta  \right)}, &\!\!\!\cdots 
            \end{array}} \!\!\right)^T} \succ 0 \\\nonumber
        &\Leftrightarrow {p_\theta } \succ 0.
    \end{align}
    By now, we have given the equivalent form of three constraints in problem $\textbf{P}_\textbf{0}$ with (\ref{proof:constrain-dp}), (\ref{proof:constrain-sum}) and (\ref{proof:constrain-positive}), respectively.
    
    Thus, with (\ref{proof:objective}), (\ref{proof:constrain-dp})-(\ref{proof:constrain-positive}), we obtain the whole equivalent form of the original optimization problem $\textbf{P}_\textbf{0}$:
    \begin{align*}
        \mathop {\min }\limits_{p_\theta}~ & W_1(x, x+\theta) = \sum\nolimits_k {\left| {{p_x^T} M_k {p_\theta}} \right|} \\[7pt]
        \text{s.t.}~~
        & A p_{\theta} \le b,
        ~ \left | p_{\theta} \right |  = 1,
        ~ p_{\theta} \succ 0,
    \end{align*}
    where the equivalence is proved by the objective variable, the objective function and the constraints, respectively.
    
\end{proof}


\balance

\bibliographystyle{ieeetr}
\bibliography{main-v.bib}

\begin{IEEEbiographynophoto}{Shuying Qin} 
(S'22) is currently an undergraduate in the Department of Automation, Shanghai Jiao Tong University, Shanghai, China.
Her research interests include privacy and security in network systems.
\end{IEEEbiographynophoto}

\begin{IEEEbiographynophoto}{Jianping He}
(SM'19) is currently an associate professor in the Department of Automation at Shanghai Jiao Tong University.
He received the Ph.D. degree in control science and engineering from Zhejiang University, Hangzhou, China, in 2013, and had been a research fellow in the Department of Electrical and Computer Engineering at University of Victoria, Canada, from Dec. 2013 to Mar. 2017.
His research interests mainly include the distributed learning, control and optimization, security and privacy in network systems.

Dr. He serves as an Associate Editor for IEEE Trans. on Control of Network Systems, IEEE Open Journal of Vehicular Technology and KSII Trans. Internet and Information Systems.
He was also a Guest Editor of IEEE TAC, IEEE TII, International Journal of Robust and Nonlinear Control, etc.
He was the winner of Outstanding Thesis Award, Chinese Association of Automation, 2015. He received the best paper award from IEEE WCSP'17, the best conference paper award from IEEE PESGM'17, the finalist best student paper award from IEEE ICCA'17, and the finalist best conference paper award from IEEE VTC'20-Fall.
\end{IEEEbiographynophoto}

\begin{IEEEbiographynophoto}{Chongrong Fang}
(M'21)
is currently an Assistant Professor with the Department of Automation, Shanghai Jiao Tong University, Shanghai, China.
He received the B.Sc. degree in automation and the Ph.D. degree in control science and engineering from Zhejiang University, Hangzhou, China, in 2015 and 2020, respectively.
His research interests include anomaly detection and diagnosis in cyber-physical systems and cloud networks.
\end{IEEEbiographynophoto}

\begin{IEEEbiographynophoto}{James Lam}
received a B.Sc. (1st Hons.) degree in Mechanical Engineering from the University of Manchester, and was awarded the Ashbury Scholarship, the A.H. Gibson Prize, and the H. Wright Baker Prize for his academic performance. He obtained the MPhil and Ph.D. degrees from the University of Cambridge. He is a Croucher Scholar, Croucher Fellow, and Distinguished Visiting Fellow of the Royal Academy of Engineering, and Cheung Kong Chair Professor. Prior to joining the University of Hong Kong in 1993 where he is now Chair Professor of Control Engineering, he was a faculty member at the City University of Hong Kong and the University of Melbourne.

Professor Lam is a Chartered Mathematician (CMath), Chartered Scientist (CSci), Chartered Engineer (CEng), Fellow of Institute of Electrical and Electronic Engineers (FIEEE), Fellow of Institution of Engineering and Technology (FIET), Fellow of Institute of Mathematics and Its Applications (FIMA), Fellow of Institution of Mechanical Engineers (FIMechE), and Fellow of Hong Kong Institution of Engineers (FHKIE).

He is Editor-in-Chief of IET Control Theory and Applications, Journal of The Franklin Institute and Proc. IMechE Part I: Journal of Systems and Control Engineering, Subject Editor of Journal of Sound and Vibration, Editor of Asian Journal of Control, Senior Editor of Cogent Engineering, Section Editor of IET Journal of Engineering, Consulting Editor of International Journal of Systems Science, Associate Editor of Automatica and Multidimensional Systems and Signal Processing.

His research interests include model reduction, robust synthesis, delay, singular systems, stochastic systems, multidimensional systems, positive systems, networked control systems and vibration control. He is a Highly Cited Researcher in Engineering (2014, 2015, 2016, 2017, 2018, 2019, 2020) and Computer Science (2015).
\end{IEEEbiographynophoto}

\end{document}